\documentclass[journal]{IEEEtran}
\IEEEoverridecommandlockouts

\usepackage{cite}
\usepackage{url}
\usepackage{amsmath}
\usepackage{amssymb}
\usepackage{amsthm}
\usepackage{mathtools}
\usepackage{graphicx}
\usepackage{booktabs}
\usepackage{bbm}
\usepackage{physics}

\newtheorem{theorem}{Theorem}
\newtheorem{lemma}{Lemma}
\newtheorem{remark}{Remark}
\newtheorem{problem}{Problem}

\title{ReLU Networks for Model Predictive Control: Network Complexity and Performance Guarantees\\\large Extended Version with Numerical Experiments}
\author{Xingchen Li, Keyou You%
\thanks{This work was supported by the National Natural Science Foundation of China (62325305), and the BNRist project (No. BNR2024TD03003). (Corresponding author: Keyou You.)}%
\thanks{Xingchen Li and Keyou You are with the Department of Automation, and Beijing National Research Center for Info. Sci. \& Tech. (BNRist), Tsinghua University, Beijing 100084, China. (e-mail: lixc21@mails.tsinghua.edu.cn, youky@tsinghua.edu.cn)}%
}
\date{\today}

\begin{document}
\maketitle
\thispagestyle{empty}
\pagestyle{empty}

\begin{abstract}
Recent years have witnessed a resurgence in using ReLU neural networks (NNs) to represent model predictive control (MPC) policies. However, determining the sufficient network complexity to ensure closed-loop guarantees remains a fundamental open problem. This involves a critical precision-complexity trade-off: undersized networks may fail to capture the MPC policy, while oversized ones may outweigh the benefits of ReLU network approximation. In this work, we project the MPC actions onto a tightened constraint set to satisfy the hard constraints, and establish a state-dependent Lipschitz continuity property for the optimal MPC cost function, which enables sharp convergence analysis of the closed-loop system. For the first time, we derive explicit bounds on ReLU network width and depth for approximating MPC policies with closed-loop guarantees. To reduce network complexity for the convergence to user-specified target sets, we propose a non-uniform error framework with a state-aware scaling function to adaptively adjust both the input and output of the ReLU network. Our contributions provide a foundational step toward certifiable ReLU NN-based MPC.
\end{abstract}


\section{Introduction}
Model predictive control (MPC) has been widely adopted in industrial applications, including chemical processes \cite{chen2012distributed}, power systems \cite{venkat2008distributed}, robotics \cite{nubert2020safe}, and autonomous vehicles \cite{cheng2019longitudinal}. However, it requires solving an optimization problem in each sampling interval, posing challenges for embedded platforms with limited computational resources \cite{nguyen2024tinympc}. Beyond developing efficient embedded solvers~\cite{stellato2020osqp,ferreau2014qpoases}, an alternative is to shift computation from \textit{online} to \textit{offline}.

One prominent approach  is explicit MPC, which leverages the solution structure of multi-parametric quadratic programs (mpQPs) and precomputes affine MPC policies for each critical region (CR)  \cite{bemporad2002explicit}.
Online, it simply identifies the CR containing the current state and applies the corresponding control action. Although techniques such as region merging \cite{geyer2008optimal} and elimination of saturated regions \cite{kvasnica2012clipping,kvasnica2013complexity} can reduce the number of CRs, this number still grows exponentially with the prediction horizon, input dimension, and number of constraints \cite{alessio2009survey,bemporad2021explicit}.

Another appealing approach directly approximates the MPC policy using kernel-based, barycentric, and orthogonal search tree methods \cite{tokmak2023automatic,jones2010polytopic,johansen2003approximate}. These, however, are generally applied to low-dimensional systems due to the curse of dimensionality. Recent advances in neural networks (NNs) have spurred growing interest in using them to approximate MPC policies. Various network architectures, e.g., feed-forward NNs, RNNs, and LSTMs \cite{nubert2020safe,quan2019approximate,kumar2018deep}, have been used, and NNs for MPC have been deployed on embedded platforms, such as DSPs \cite{lu2014convergence}, FPGAs \cite{wang2021model,dong2023standoff}, and ARM CPUs \cite{dong2023standoff}, and applied to resonant power converters, UAVs, and robotic manipulators \cite{lucia2020deep,dong2023standoff,nubert2020safe}.

Research on NN-based MPC focuses on enhancing sample efficiency and guaranteeing the closed-loop behavior. For efficient sample generation, methods such as nonlinear programming sensitivity-based data augmentation \cite{krishnamoorthy2021sensitivity} and the geometric random walk algorithm \cite{chen2022large} have been proposed. To preserve feasibility or stability of the original MPC policy, techniques such as the output projection layer \cite{chen2018approximating}, robust MPC formulations \cite{hertneck2018learning,nubert2020safe,tokmak2023automatic}, and post-approximation optimization \cite{chen2022large} are employed. In particular, projecting the network output onto suitable safe sets can guarantee robust feasibility and input-to-state stability \cite{paulson2020approximate}. Stability of NN-based MPC policies can also be verified deterministically via mixed-integer programs \cite{fabiani2022reliably,schwan2023stability} or probabilistically via sampling \cite{hertneck2018learning}; sampling-based methods can further certify the feasibility and near-optimality of the learned policy \cite{zhang2020near}. A review of NN-based MPC is available in \cite{gonzalez2023neural}.

Despite the rapidly expanding body of literature, a foundational question regarding the sufficient network complexity (i.e., depth and width) for NN-based MPC policies to achieve closed-loop guarantees remains open \cite{karg2023efficient,fabiani2022reliably}. This highlights a core precision-complexity trade-off: while an undersized network may fail to capture the MPC policy, an oversized one can incur high inference latency that may outweigh the benefits of NN approximation. Exact representation of explicit MPC laws \cite{karg2020efficient,fahandezh2020proximity} and other piecewise affine control laws \cite{fabiani2024robust} does give conditions on the network complexity, but they are tied to the number of CRs, which typically grows exponentially and is not easy to specify\cite{alessio2009survey}. Moreover, in practice the network is trained to approximate the MPC policy rather than to represent it exactly. For approximation, the available guarantees apply to a given trained network \cite{fabiani2022reliably,schwan2023stability}, rely on an online safeguard \cite{paulson2020approximate}, or are probabilistic \cite{hertneck2018learning,zhang2020near}, and, to our knowledge, none of them provides a bound on the sufficient network complexity. Existing results on ReLU approximation theory offer relevant insights into this question. Shen et al. \cite{shen2022optimal} established {\em uniform} error bounds for ReLU networks approximating continuous functions. Specifically, for any continuous function $f:[0, 1]^n\rightarrow \mathbb{R}$, there exists a ReLU network achieving a uniform approximation error bound of $\mathcal{O}(\sqrt{n}\,\omega_f((N^2L^2\ln N)^{-1/n}))$, where the network has width $\mathcal{O}(N)$ and depth $\mathcal{O}(L)$, and $\omega_f(\cdot)$ is a continuity modulus of $f$. Their construction attains the minimal error bound in big $\mathcal{O}$-notation and is shown to be tight via Vapnik--Chervonenkis (VC)-dimension arguments. ReLU networks can also approximate real analytic functions \cite{opschoor2022exponential} and those in Besov spaces \cite{ali2021approximation}. Further results are summarized in recent reviews \cite{berner2021modern,devore2021neural}.

Nevertheless, directly applying those results to approximate MPC policies  faces two fundamental hurdles. First,  approximating the MPC policy up to a prescribed error does not by itself preserve the hard constraints of the MPC formulation. When the MPC policy operates on constraint boundaries, even minor approximation errors can lead to constraint violations. We therefore project the output of MPC policies onto a tightened constraint set before approximation. If the gap between the original and tightened sets is sufficient to encompass the approximation error bound, constraint satisfaction can be guaranteed. Whereas robust MPC formulations \cite{hertneck2018learning,nubert2020safe,tokmak2023automatic} account for the approximation error in the design of the MPC itself, the projection leaves the given MPC policy unchanged and accounts for this error in the approximation target. In contrast to projecting the network output online \cite{paulson2020approximate}, this projection is applied offline to the training targets and adds no online computation.

Second, approximation errors can compromise closed-loop performance and even provoke instability. We establish a novel state-dependent Lipschitz continuity property for the optimal MPC cost function, instead of the commonly used global Lipschitz condition \cite{borrelli2017predictive}. This property captures the local behavior more accurately, providing a sharper tool for closed-loop convergence analysis. Crucially, it permits approximation errors to scale linearly with the system state, motivating our design of non-uniform error bounds for ReLU networks in Section IV. We accordingly propose state-aware scaling functions that adaptively adjust the input and output of the network to meet these desired error bounds. Whereas localized kernel interpolation \cite{tokmak2023automatic} varies the sampling density and the kernel length scale over the state space to meet the same error bound everywhere, we let the required accuracy itself vary with the state. Our approach enhances closed-loop performance under the ReLU NN-based MPC policy and can also reduce the sufficient network complexity for small target sets.

Overall,  our main contributions, along with their key takeaways, are summarized below:

\begin{itemize}
\item \textbf{Explicit bounds on ReLU network complexity:} We establish the first explicit bounds on the depth and width of ReLU networks that approximate the MPC policy up to a prescribed error and guarantee closed-loop performance. These bounds are given in terms of the mpQP dimensions and the constraint set, and do not depend explicitly on the number of CRs.
\item \textbf{Sharp tool for closed-loop convergence:} We establish a state-dependent Lipschitz continuity property in Lemma \ref{lem:3Lip} for the optimal MPC cost function. Under a uniform error bound, it guarantees closed-loop convergence to an invariant set whose size scales linearly with the error bound instead of with its square root as in \cite{fabiani2022reliably}, consistent with the scaling invariance of linear systems.
\item \textbf{Non-uniform error bounds for ReLU network approximation:} Whereas existing ReLU approximation theory provides uniform error bounds, the state-aware scaling functions let a static ReLU network meet a prescribed state-dependent error bound. As the target set shrinks, the resulting complexity bound grows only polylogarithmically, instead of polynomially, in the inverse of its size (cf. Section \ref{sec:pro2}).
\end{itemize}

The remainder of this paper is organized as follows. In Section \ref{sec_formulation}, we formulate the problem under investigation. In Section \ref{sec:main_results}, we provide the main complexity results for ReLU NN-based MPC policies with a uniform error bound. In Section \ref{sec:pro2}, we sharpen the analysis by establishing a non-uniform error bound to enhance closed-loop performance and, for small target sets, reduce network complexity. In Section \ref{sec:simu}, we provide numerical examples. Finally, Section \ref{sec:conclusion} concludes the paper.

\textbf{Notation}:
We denote the sets of real, non-negative real, natural, and positive natural numbers by $\mathbb{R}$, $\mathbb{R}_+$, $\mathbb{N}$, and $\mathbb{N}_+$, respectively. The sets of $n \times n$ symmetric, positive semi-definite, and positive definite matrices are $\mathbb{S}^n$, $\mathbb{S}^n_+$, and $\mathbb{S}^n_{++}$. For a vector $x \in \mathbb{R}^n$, $\|x\|$ is the Euclidean norm, and for a matrix $Q \in \mathbb{S}^n_+$, $\|x\|_Q := \sqrt{x^\top Qx}$. The closed ball of radius $\epsilon$ centered at $x$ is $\mathcal{B}(x, \epsilon)$, and $\mathcal{B}(\epsilon) := \mathcal{B}(0, \epsilon)$. For a bounded set $\mathcal{S}$ containing the origin, its radius is $D(\mathcal{S}) := \sup_{x \in \mathcal{S}} \|x\|$, and its inradius is $d(\mathcal{S}) := \sup\{\delta : \mathcal{B}(\delta) \subseteq \mathcal{S}\}$. The projection of $x$ onto a closed convex set $\mathcal{S}$ is $\Pi_{\mathcal{S}}(x) := \arg\min_{y \in \mathcal{S}} \|x-y\|$. For a Lipschitz continuous function $f:\mathcal{S} \to \mathbb{R}^m$, its global Lipschitz constant (GLC)  is $L_f= \sup_{x,y \in \mathcal{S}, x \neq y} {\|f(x)-f(y)\|}/{\|x-y\|}$, and $\mathcal{L}(f,\mathcal{S}')$ denotes the GLC of $f$ restricted to $\mathcal{S}'\subseteq\mathcal{S}$.  Throughout the paper, C-sets are defined as convex and compact sets containing the origin in their interior. For a set $\mathcal{S}$, $\mathrm{\mathbf{cl}}(\mathcal{S})$ denotes its closure, $\mathrm{dist}(x,\mathcal{S}):=\inf_{y\in\mathcal{S}}\|x-y\|$ is the distance from $x$ to $\mathcal{S}$, and $\mathbbm{1}(\cdot)$ is the indicator function.

\section{Problem Formulation}\label{sec_formulation}
This section first introduces the linear model predictive control (MPC) setup  for quadratic stabilization of constrained linear time-invariant (LTI) systems and the ReLU NN-based MPC policy. Then, we formulate the problem of determining the network complexity for closed-loop guarantees.

\subsection{MPC for quadratic stabilization of constrained LTI systems}
\label{sec:mpc_stabilization}
We consider the following constrained discrete-time LTI system:
\begin{equation}\label{equ:system}
    x_{t+1}=Ax_t+Bu_t, ~t\in\mathbb{N},
\end{equation}
where $x_t\in\mathcal{X}\subseteq\mathbb{R}^{n_x}$ is the state vector, $u_t\in\mathcal{U}\subseteq\mathbb{R}^{n_u}$ is the control input, $A\in\mathbb{R}^{n_x\times n_x}$ and $B\in\mathbb{R}^{n_x\times n_u}$, $B\neq0$, are system matrices. MPC is a well-established method for stabilizing such constrained systems. At each time step $t$, it computes the control input by solving
\begin{equation}\label{equ:mpc}
\begin{aligned}
    \operatorname*{minimize}_{u_{0|t},\dots,u_{n_p-1|t}} \quad& \|x_{n_p|t}\|_P^2+\sum_{k=0}^{n_p-1}\big(\|x_{k|t}\|_Q^2+\|u_{k|t}\|_R^2\big)\\
    \text{s.t.} \quad& x_{k+1|t}=Ax_{k|t}+Bu_{k|t},\\
    &x_{k|t}\in\mathcal{X},~u_{k|t}\in\mathcal{U}, ~x_{n_p|t}\in\mathcal{X}_{f}\\
    &x_{0|t}=x_t, k=0,\dots,n_p - 1.
\end{aligned}
\end{equation}
Here, $P, Q \in \mathbb{S}^{n_x}_+$ and $R \in \mathbb{S}^{n_u}_{++}$; $n_p \in \mathbb{N}_+$ is the prediction horizon. The sets $\mathcal{X}$, $\mathcal{U}$, and the terminal set $\mathcal{X}_f$ are polytopic C-sets.

Clearly, the optimal solution to \eqref{equ:mpc} is a function of the current state vector $x_t$. We denote the first element of the optimal control sequence as $u_{0|t}^*(x_t)$ and the optimal cost as $v_\text{mpc}(x_t)$. Then, the MPC policy is given as
\begin{equation}\label{equ:mpc_policy}
u_\text{mpc}(x_t) := u_{0|t}^*(x_t).
\end{equation}

Under standard MPC design choices \cite[Sec.\ 2.5.4]{rawlings2017model}, the MPC policy \eqref{equ:mpc_policy} renders the closed-loop system \eqref{equ:system} exponentially stable. In particular, there exist a C-set $\mathcal{X}_\text{inv}$ and constants $c_1,c_2,c_3>0$ such that, for all $x \in \mathcal{X}_\text{inv}$, the following inequalities hold:
\begin{subequations}\label{equ:exp_stability1-3}
\begin{align}
&Ax + Bu_\text{mpc}(x) \in \mathcal{X}_\text{inv},\label{equ:exp_stability1}\\
&c_1 \|x\|^2 \leq v_\text{mpc}(x) \leq c_2 \|x\|^2,\label{equ:exp_stability2}\\
&v_\text{mpc}(Ax + Bu_\text{mpc}(x)) - v_\text{mpc}(x) \leq -c_3 \|x\|^2.\label{equ:exp_stability3}
\end{align}
\end{subequations}

By \eqref{equ:exp_stability1-3}, $v_\text{mpc}(\cdot)$ is a Lyapunov function on $\mathcal{X}_\text{inv}$ for the closed-loop system \eqref{equ:system} under $u_\text{mpc}(\cdot)$. For simplicity, we choose $\mathcal{X}_\text{inv}$ as a sublevel set of the optimal cost function $v_\text{mpc}(\cdot)$, i.e.,
\begin{equation}\label{sublevel}
\mathcal{X}_\text{inv}=\mathcal{V}(\gamma):=\{x{}\in\mathcal{X}:{}v_\text{mpc}(x)\le \gamma\}.\end{equation} Here, $\gamma<\gamma_{\max}$, where $\gamma_{\max}$ is the supremum of the levels for which $\mathcal{V}(\gamma)$ lies in the interior of the feasible set of \eqref{equ:mpc}, i.e., $\mathcal{F}:=\{x\in\mathcal{X}:\eqref{equ:mpc}\text{ is feasible with }x_t=x\}$.
Note that the constants $c_1, c_2, c_3$ and $\gamma_{\max}$ can be computed by leveraging the piecewise affine structure of the MPC policy \cite{borrelli2017predictive}.

\subsection{The ReLU NN-based MPC policy}
Since the MPC policy \eqref{equ:mpc_policy} is piecewise affine, it is natural to approximate it with a ReLU neural network (NN), yielding a ReLU NN-based MPC policy:
\begin{equation}\label{equ:nn}
    u_\text{nn}(x) = l_{n_d} \circ \sigma  \circ \cdots \circ l_2 \circ \sigma \circ l_1(x),
\end{equation}
where $l_i(x):= W_i x + b_i$, $W_i\in\mathbb{R}^{h_{i+1}\times h_i}$, $b_i\in\mathbb{R}^{h_{i+1}}$, $i=1,\dots,n_d$, and $\sigma(\cdot):=\max\{0,\cdot\}$ is the vectorized ReLU activation function. Then, the closed-loop system is given by
\begin{equation}\label{equ:closed_loop}
    x_{t+1}=Ax_t+Bu_\text{nn}(x_t).
\end{equation}

Indeed, the idea of using ReLU NNs to approximate MPC policies is not new and has been widely adopted, see e.g., \cite{karg2020efficient,dong2023standoff,lucia2020deep,gonzalez2023neural}. Nevertheless, determining the network complexity (i.e., width $n_w:=\max_{i=2,\ldots,n_d}\{h_i\}$ and depth $n_d$) sufficient for closed-loop guarantee of \eqref{equ:closed_loop} remains an open problem. This work aims to bridge this gap.

\subsection{Objective of this work}
In this work, we focus on the constrained LTI system \eqref{equ:system} together with the MPC scheme \eqref{equ:mpc} satisfying \eqref{equ:exp_stability1-3}. As in \cite{fabiani2022reliably,schwan2023stability}, we first examine the case with a {\em uniform} error bound.

\begin{problem}[Uniform error bound]\label{problem1} Determine an upper bound $\bar\delta$ and, for any $\delta_1 \in (0, \bar\delta]$, characterize sufficient network complexity $n_w, n_d \in \mathbb{N}_+$ for which there exists a ReLU NN-based MPC policy of the form \eqref{equ:nn} satisfying, for every $x\in\mathcal{X}_\text{inv}$,
\begin{subequations}
\label{equ:con}
\begin{align}
    &u_\text{nn}(x)\in\mathcal{U},\label{equ:con2}\\
    &\|u_\text{nn}(x)-u_\text{mpc}(x)\|\leq\delta_1,\label{equ:con1}\\
    &Ax+Bu_\text{nn}(x)\in\mathcal{X}_\text{inv}. \label{equ:con3}
\end{align}
\end{subequations}
Moreover, for any $x_0\in\mathcal{X}_\text{inv}$, the closed-loop system \eqref{equ:closed_loop} converges exponentially to an invariant set of size $\mathcal{O}(\delta_1)$ contained in $\mathcal{X}_\text{inv}$.
\end{problem}

Prior works \cite{fabiani2022reliably,schwan2023stability} first assume the availability of a ReLU network satisfying \eqref{equ:con} and then certify the closed-loop behavior of a given trained network a posteriori, but do not quantify the sufficient network width and depth. Under this condition, the analysis of \cite{fabiani2022reliably} yields convergence of the closed-loop system to an invariant set of size $\mathcal{O}(\sqrt{\delta_1})$. In addition to establishing explicit bounds for the width and depth of ReLU networks, we sharpen the closed-loop convergence to a smaller invariant set of size $\mathcal{O}(\delta_1)$, which is consistent with scale-invariance of linear systems.

Then, from the perspective of {\em non-uniform} error bounds, we further improve the closed-loop convergence to an arbitrarily small user-chosen target C-set $\mathcal{X}_\text{inv}'$ while reducing the network complexity when $\mathcal{X}_\text{inv}'$ is small. To the best of our knowledge, state-dependent error bounds have not been proposed for designing ReLU NN-based MPC policies.

\begin{problem}[Non-uniform error bounds]\label{problem2} For any C-set $\mathcal{X}_\text{inv}'\subseteq\mathcal{X}_\text{inv}$, design a non-uniform error bound function $\delta_2:\mathbb{R}^{n_x}\to\mathbb{R}_+$ and characterize sufficient network complexity $n_w,n_d\in\mathbb{N}_+$ for which there exists an NN-based MPC policy $u_\text{nn}(\cdot)$ satisfying, for every $x\in\mathcal{X}_\text{inv}$,
\begin{subequations}\label{equ:2con}
\begin{align}
    &u_\text{nn}(x)\in\mathcal{U},\label{equ:2con2}\\
    &\|u_\text{nn}(x)-u_\text{mpc}(x)\|\leq\delta_2(x),\label{equ:2con1}\\
    &Ax+Bu_\text{nn}(x)\in\mathcal{X}_\text{inv},\label{equ:2con3}
\end{align}
\end{subequations}
and the closed-loop system \eqref{equ:closed_loop} starting from any $x_0\in\mathcal{X}_\text{inv}$ converges exponentially to the target set $\mathcal{X}_\text{inv}'$.
\end{problem}

\begin{figure*}[thbp]
\centering
\includegraphics[width=1\linewidth]{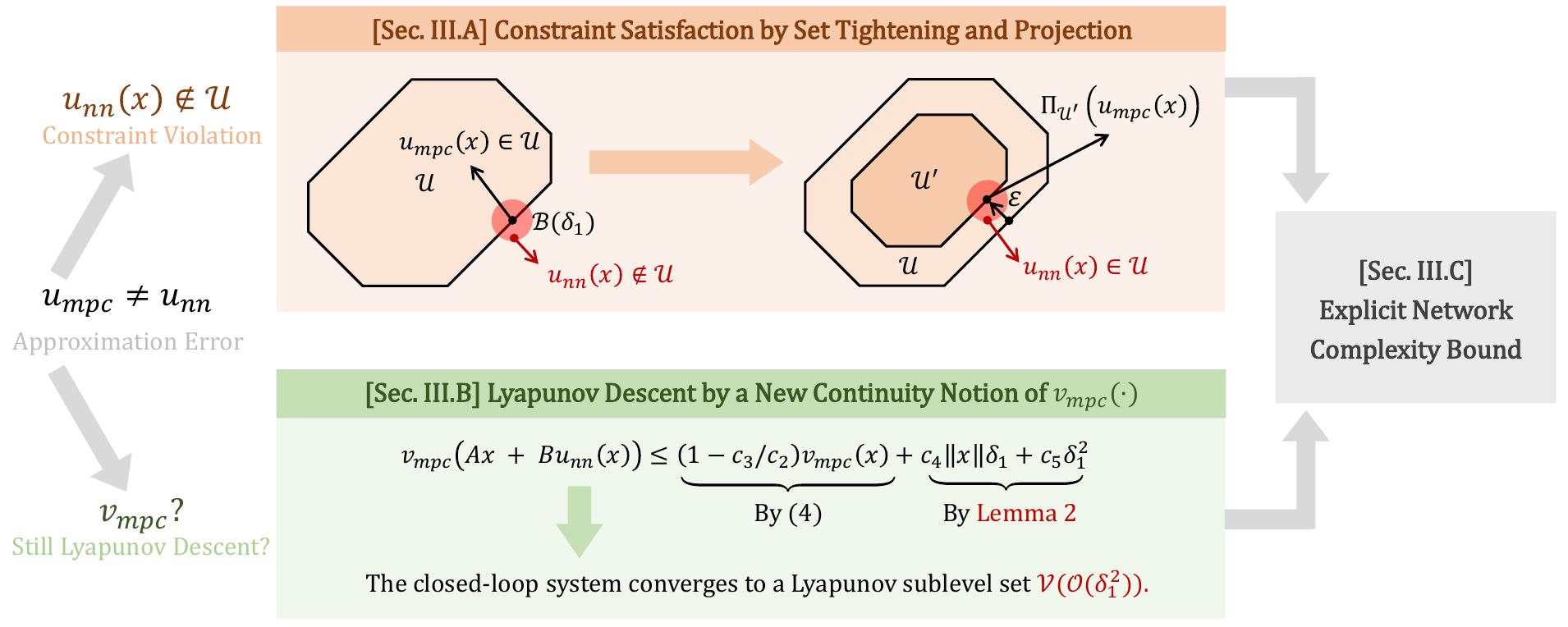}
\vspace{-2em}
\caption{Illustration of the main challenges and our analytical approach. }
\vspace{-1em}
\label{fig:constraint_issue}
\end{figure*}

A key distinction between \eqref{equ:2con1} and the uniform error bound in \eqref{equ:con1} is the use of $\delta_2(\cdot)$ as a designable error bound function. This allows the approximation error to vary with the state of the closed-loop system. Crucially, a properly designed error bound function $\delta_2(\cdot)$ can not only improve the closed-loop convergence to an arbitrarily small C-set but also reduce the sufficient network complexity. In particular, we explicitly show that, for small target sets $\mathcal{X}_\text{inv}'$, the sufficient network complexity admits a smaller bound than under the uniform error bound. To realize \eqref{equ:2con1}, we design novel state-aware scaling functions that adaptively modulate the input and output of the ReLU network.
\section{Network Complexity for the ReLU NN-based MPC Policy in Solving Problem \ref{problem1}}\label{sec:main_results}

In this section, we establish the explicit  complexity of the ReLU network sufficient to solve Problem \ref{problem1}. The core challenge lies in managing approximation errors between the MPC policy in \eqref{equ:mpc_policy}  and the ReLU NN-based MPC policy in \eqref{equ:nn}, which may violate the control input constraint in \eqref{equ:con2} and the descent property in \eqref{equ:exp_stability3}.

Our approach, illustrated in Fig. \ref{fig:constraint_issue}, addresses these issues and derives the final complexity bounds in three main steps:
\begin{itemize}
    \item \textbf{Constraint satisfaction via projection:} To prevent constraint violations, we first project the MPC policy onto a tightened constraint set, so that the output of the ReLU network remains in $\mathcal{U}$ for approximation errors up to a prescribed bound. See Section \ref{sec:projection}.
    \item \textbf{Sharpened closed-loop convergence:} To study the impact of approximation errors on closed-loop convergence, we establish a state-dependent Lipschitz continuity property of the optimal MPC cost function $v_\text{mpc}(\cdot)$ in Section \ref{sec:lip}. It allows a sharper analysis than the global Lipschitz condition, which shows that the closed-loop system converges to an invariant set of size $\mathcal{O}(\delta_1)$.
    \item \textbf{Complexity bounds of ReLU networks:} Building on the above guarantees and recent advances in ReLU network approximation theory \cite{shen2022optimal}, we establish explicit bounds on the network width and depth to solve Problem \ref{problem1} in Section \ref{sec:complexity}.
\end{itemize}

We now detail the proposed approach.

\subsection{Constraint satisfaction via set tightening and projection}\label{sec:projection}
As illustrated in Fig.~\ref{fig:constraint_issue}, if $u_\text{mpc}(x)$ lies on the boundary of $\mathcal{U}$, even an arbitrarily small approximation error can drive $u_\text{nn}(x)$ outside $\mathcal{U}$.

To ensure constraint satisfaction, we project the control action \(u_{\text{mpc}}(x)\) onto a tightened constraint set:
\begin{equation}\label{tightset}
\mathcal{U}'(\epsilon) := \{ u \in \mathcal{U} : u + \Delta \in \mathcal{U}, \; \forall \Delta \in \mathcal{B}(\epsilon) \}
\end{equation}
and shift our approximation target to the following tightened counterpart
\begin{equation}\label{tightproj}
\Pi_{\mathcal{U}'(\epsilon) }(u_{\text{mpc}}(x))\in \arg\min_{u\in\mathcal{U}'(\epsilon) }\|u-u_{\text{mpc}}(x)\|.\end{equation}

Since $\mathcal{U}'(\epsilon) $ is a C-set for every $\epsilon\in[0,d(\mathcal{U}))$ (cf. Lemma~\ref{lem:3Lip1}), the above projection is well-defined. If the gap between $\mathcal{U}$ and $\mathcal{U}'(\epsilon) $ is sufficient to encompass approximation errors, constraint satisfaction can be guaranteed, which is formalized in Lemma~\ref{lem:projection}.

\begin{lemma}\label{lem:projection}
Let $\epsilon \le \delta_1 d(\mathcal{U})/(2D(\mathcal{U}))$ and $\epsilon<d(\mathcal{U})$. If
\begin{equation}\label{tightedapprox}
\|u_{\text{nn}}(x)-\Pi_{\mathcal{U}'(\epsilon) }(u_{\text{mpc}}(x))\| \le \epsilon, \forall x \in \mathcal{X}_\text{inv},
\end{equation}
 then $\| u_{\text{nn}}(x) - u_{\text{mpc}}(x) \| \leq \delta_1$ and $u_{\text{nn}}(x) \in \mathcal{U}$ for every $x\in\mathcal{X}_\text{inv}$, i.e., \eqref{equ:con2} and \eqref{equ:con1} hold.
\end{lemma}
\begin{proof}
Define the maximum projection distance
$$r(\mathcal{U},\epsilon) := \max_{u \in \mathcal{U}} \|u - \Pi_{\mathcal{U}'(\epsilon) }(u)\|.$$
By the triangle inequality and Lemma \ref{lem:3Lip2} in Appendix \ref{sec:proof_3Lip2}, we obtain that
\begin{align*}
\|u_{\text{nn}}(x)-u_{\text{mpc}}(x)\| & \le  \|u_{\text{nn}}(x)-\Pi_{\mathcal{U}'(\epsilon) }(u_{\text{mpc}}(x))\| \\
&~~~+ \|\Pi_{\mathcal{U}'(\epsilon) }(u_{\text{mpc}}(x)) - u_{\text{mpc}}(x)\| \\
& \le \epsilon + r(\mathcal{U}, \epsilon) \leq \epsilon + \epsilon D(\mathcal{U})/d(\mathcal{U}) \\
&\le \delta_1,
\end{align*}
where the last step uses $\epsilon\le\delta_1 d(\mathcal{U})/(2D(\mathcal{U}))$ and $d(\mathcal{U})\le D(\mathcal{U})$.

Since $\Pi_{\mathcal{U}'(\epsilon) }(u_{\text{mpc}}(x)) \in \mathcal{U}'(\epsilon) $ and $u_{\text{nn}}(x)-\Pi_{\mathcal{U}'(\epsilon) }(u_{\text{mpc}}(x))\in \mathcal{B}(\epsilon)$ (cf. \eqref{tightedapprox}), it follows from \eqref{tightset} that $u_{\text{nn}}(x) \in \mathcal{U}$.
\end{proof}
Consequently, for a given $\delta_1<2D(\mathcal{U})$, the objective of the ReLU network with respect to \eqref{equ:con2} and \eqref{equ:con1} is reformulated as \eqref{tightedapprox} with $\epsilon=\delta_1 d(\mathcal{U})/(2D(\mathcal{U}))$, the largest value allowed by Lemma~\ref{lem:projection}. Building on state-of-the-art results on approximation by ReLU networks \cite{shen2022optimal}, we derive the bound on network complexity in Section \ref{sec:complexity}.

\begin{remark}\label{rem:training}
The ReLU NN-based MPC policy \eqref{equ:nn} can be obtained through supervised learning. A training set is constructed by sampling state vectors $x \in \mathcal{X}_\text{inv}$, solving the MPC problem \eqref{equ:mpc} (with $x_t = x$) to get $u_{\text{mpc}}(x)$, and applying the projection in \eqref{tightproj}, which is again a convex quadratic program (QP), to obtain the target control input. This dataset of state-target pairs is then used to train the ReLU network. How to train a network that satisfies \eqref{tightedapprox} is a general challenge in machine learning and is beyond the scope of this work.
\end{remark}

\subsection{State-dependent Lipschitz continuity  of  $v_\text{mpc}(\cdot)$ for improved closed-loop convergence}\label{sec:lip}
In the presence of approximation errors that do not vanish, the closed-loop system \eqref{equ:closed_loop} may not converge to the origin. Instead, it has been established in \cite{fabiani2022reliably} under \eqref{equ:con1} by exploiting the global Lipschitz continuity of the optimal MPC cost function $v_\text{mpc}(\cdot)$ \cite[Remark 5.8]{borrelli2017predictive}, \cite[Prop.\ 7.13]{rawlings2017model}, i.e.,
\begin{equation}\label{glc}
|v_\text{mpc}(x_1) - v_\text{mpc}(x_2)| \leq {L}_{v}\cdot \|x_1 - x_2\|
\end{equation}
where ${L}_{v}$ is the global Lipschitz constant (GLC).  Combined with \eqref{equ:con2} and \eqref{equ:con1}, it follows that
\begin{align}\label{lyapunov}
|v_\text{mpc}(Ax + Bu_\text{nn}(x)) - v_\text{mpc}(Ax + Bu_\text{mpc}(x))| \le  {L}_{v} \|B\| \delta_1.
\end{align}
Together with the decrease condition in \eqref{equ:exp_stability1-3}, this yields
\begin{equation}\label{traditional_result}
v_\text{mpc}(Ax + Bu_\text{nn}(x)) \le (1 - c_3/c_2) v_\text{mpc}(x) + \delta_1{L}_{v}  \|B\|.
\end{equation}

Consequently, the closed-loop system \eqref{equ:closed_loop} converges to an invariant set of size $\mathcal{O}\left(\sqrt{\delta_1  {L}_{v}  \|B\|{}\cdot c_2 / (c_1 c_3)} \right)$. However, this conclusion suffers from two limitations. First, the GLC of $v_{\text{mpc}}(\cdot)$ can become prohibitively large. Second, the  size of the resulting invariant set scales as $\mathcal{O}(\sqrt{\delta_1})$, which is inconsistent with the scaling invariance of LTI systems; a sharper analysis should yield an invariant set of size $\mathcal{O}(\delta_1)$.

To overcome these limitations, we establish a state-dependent Lipschitz continuity property of $v_\text{mpc}(\cdot)$, which captures its behavior more accurately. That is, the GLC of \eqref{glc} can be replaced with a state-dependent Lipschitz bound that scales linearly with the norm of the state vector.

\begin{lemma}[State-dependent Lipschitz continuity of $v_\text{mpc}$]\label{lem:3Lip}
Consider the MPC formulation in \eqref{equ:mpc} satisfying \eqref{equ:exp_stability1-3}. There exists a constant $c_0 > 0$ such that, for all $x_1,x_2\in\mathcal{F}$,
\begin{align}\label{newlc}
|v_\text{mpc}(x_1) - v_\text{mpc}(x_2)|& \le c_0 \max\{\|x_1\|, \|x_2\|\} \|x_1 - x_2\|.
\end{align}
\end{lemma}

\begin{proof}
See Appendix \ref{sec:proof_3Lip}.
\end{proof}

\begin{remark}The constant $c_0$ can be computed as
\begin{equation*}
c_0=2\Big(L_{x,n_p}^2\|P\|+\sum_{i=0}^{n_p-1}\big(L_{x,i}^2\|Q\|+L_{u,i}^2\|R\|\big)\Big),
\end{equation*}
where, for each prediction step $i$, $L_{u,i}$ and $L_{x,i}$ are Lipschitz constants of the optimal input $u^*_{i|t}(\cdot)$ and the optimal predicted state $x^*_{i|t}(\cdot)$ over the feasible set $\mathcal{F}$. Hence, $c_0$ depends only on $A$, $B$, $P$, $Q$, $R$, $n_p$, $\mathcal{X}$, $\mathcal{U}$ and $\mathcal{X}_f$.
\end{remark}
In comparison with the GLC in \eqref{glc}, the continuity bound in \eqref{newlc} is tighter near the origin, enabling a sharper closed-loop analysis. To elaborate,
we combine Lemma \ref{lem:3Lip}  with  \eqref{equ:con1} to obtain
\begin{equation}\label{equ:lemma2-part}
\begin{aligned}
&v_\text{mpc}(Ax + Bu_\text{nn}(x)) - v_\text{mpc}(Ax + Bu_\text{mpc}(x))\\
&\leq c_0 \max\{\|Ax + Bu_\text{nn}(x)\|, \|Ax + Bu_\text{mpc}(x)\|\} \|B\|\delta_1\\
&\leq c_0 (\|A\| \|x\| + \|B\| (\|u_\text{mpc}(x)\| + \delta_1)) \|B\| \delta_1\\
&\leq c_0 (\|A\| \|x\| + \|B\| (L_u\|x\|+ \delta_1)) \|B\|\delta_1\\
& = c_4 \|x\| \delta_1 + c_5 \delta_1^2,
\end{aligned}
\end{equation}
where $c_4 = c_0 \left(\|A\| + L_u\|B\|\right)\|B\|$, $c_5 = c_0 \|B\|^2$ and $L_u$${}=\mathcal{L}(u_\text{mpc},\mathcal{X}_\text{inv})$ is the GLC of $u_\text{mpc}(\cdot)$ over $\mathcal{X}_\text{inv}$. Together with \eqref{equ:exp_stability1-3}, we obtain
\begin{equation} \label{newlyapunova}
\begin{aligned}
v_\text{mpc}(Ax + Bu_\text{nn}(x)) &\le  (1 - c_3/c_2) v_\text{mpc}(x) \\
&~~~+ c_4 \|x\| \delta_1 + c_5 \delta_1^2.
\end{aligned}
\end{equation}

In contrast to \eqref{traditional_result}, the second term on the right-hand side (RHS) of \eqref{newlyapunova} now depends linearly on $\|x\|$, while the third term exhibits a quadratic dependence on $\delta_1$. We now show that the closed-loop system converges to a smaller invariant set.

\begin{lemma}\label{lem:convergence}
If $c_6\delta_1^2<\gamma$ and the closed-loop system of \eqref{equ:closed_loop} satisfies \eqref{newlyapunova} for all $x\in\mathcal{X}_\text{inv}$, then, for any $x_0\in\mathcal{X}_\text{inv}$, it converges exponentially to the invariant set $\mathcal{V}(c_6 \delta_1^2)$, where $c_6$ is a positive constant explicitly depending on $c_0,\ldots,c_5$.
\end{lemma}
\begin{proof}
See Appendix \ref{sec:proofL3}.
\end{proof}

Specifically, $v_\text{mpc}(x_k)-c_6\delta_1^2\le q^k\gamma$ for all $k\in\mathbb{N}$, where $x_k$ is the closed-loop state at time $k$ and $q:=1-c_3/(2c_2)-c_5/(2c_6)\in[0,1)$ does not depend on $\delta_1$. By Lemma~\ref{lem:convergence} and \eqref{equ:exp_stability2}, the closed-loop system \eqref{equ:closed_loop} converges to an invariant set of size $\mathcal{O}(\delta_1)$, compared with $\mathcal{O}(\sqrt\delta_1)$ in \cite{fabiani2022reliably} and our previous work \cite{li2025relu}. The setting is the same as in \cite{fabiani2022reliably}, namely an LTI system, the MPC cost as a Lyapunov function and a uniform error bound, so the reduction from $\mathcal{O}(\sqrt{\delta_1})$ to $\mathcal{O}(\delta_1)$ is attributable to the state-dependent Lipschitz continuity of $v_\text{mpc}(\cdot)$ developed in Lemma~\ref{lem:3Lip}.

\subsection{Explicit bounds on the sufficient ReLU network complexity}\label{sec:complexity}

Finally, we establish explicit complexity bounds on the width and depth for the ReLU NN-based MPC policy in terms of the error bound $\delta_1$ and the parameters of the MPC formulation in \eqref{equ:mpc}, including $n_x$, $n_u$, $D(\mathcal{U})$, and $d(\mathcal{U})$.

\begin{theorem}[Network Complexity for Problem \ref{problem1} and Closed-Loop Convergence]\label{thm:main}
Consider Problem \ref{problem1} with $\bar\delta = \min\{c_7d(\mathcal{U})\sqrt{\gamma}/(2 D(\mathcal{U})), d(\mathcal{U})\}$ and any $\delta_1\in(0,\bar\delta]$. If there is a pair of positive integers $(n_w',n_d')$ satisfying
\begin{equation}\label{equ:main}
\begin{aligned}
&n_w'^2n_d'^2\log_3(n_w'+2)\\
&\geq~\left(\frac{524\sqrt{n_x n_u}\,D(\mathcal{U})\,D(\mathcal{X}_\text{inv})\,{L}_u}{\delta_1\,d(\mathcal{U})}\right)^{n_x},
\end{aligned}
\end{equation}
there exists a ReLU NN-based MPC policy \eqref{equ:nn} with
\begin{subequations}
\begin{align}
n_w&=n_u3^{n_x+3}\max\{n_x\lfloor n_w'^{1/n_x}\rfloor,n_w'+2\},\label{equ:width}\\
n_d&=11n_d'+19+2n_x,\label{equ:depth}
\end{align}
\end{subequations}
such that \eqref{equ:con} holds. Moreover, for any $x_0\in\mathcal{X}_\text{inv}$, the closed-loop system \eqref{equ:closed_loop} converges exponentially to the invariant set $\mathcal{V}(c_6 \delta_1^2)$. Here, $c_0,\ldots,c_6$ are given in Lemma \ref{lem:3Lip}, \eqref{equ:exp_stability1-3}, \eqref{equ:lemma2-part} and the proof of Lemma \ref{lem:convergence}, and $c_7$, depending on $c_0,\ldots,c_5$, is given in Appendix \ref{sec:proof_main}.
\end{theorem}
\begin{proof}
See Appendix \ref{sec:proof_main}.
\end{proof}

Theorem \ref{thm:main} admits multiple feasible pairs of $(n_w',n_d')$. For instance, setting the RHS of \eqref{equ:main} to $100$ yields a simplified inequality $n_w'^2n_d'^2\log_3(n_w'+2)\ge 100$. The corresponding feasible region is illustrated by the green-shaded area in Fig. \ref{fig:NL}, where each black point marks a valid pair of $(n_w', n_d')$.

\begin{figure}[tbp]
\centering
\includegraphics[width=1\linewidth]{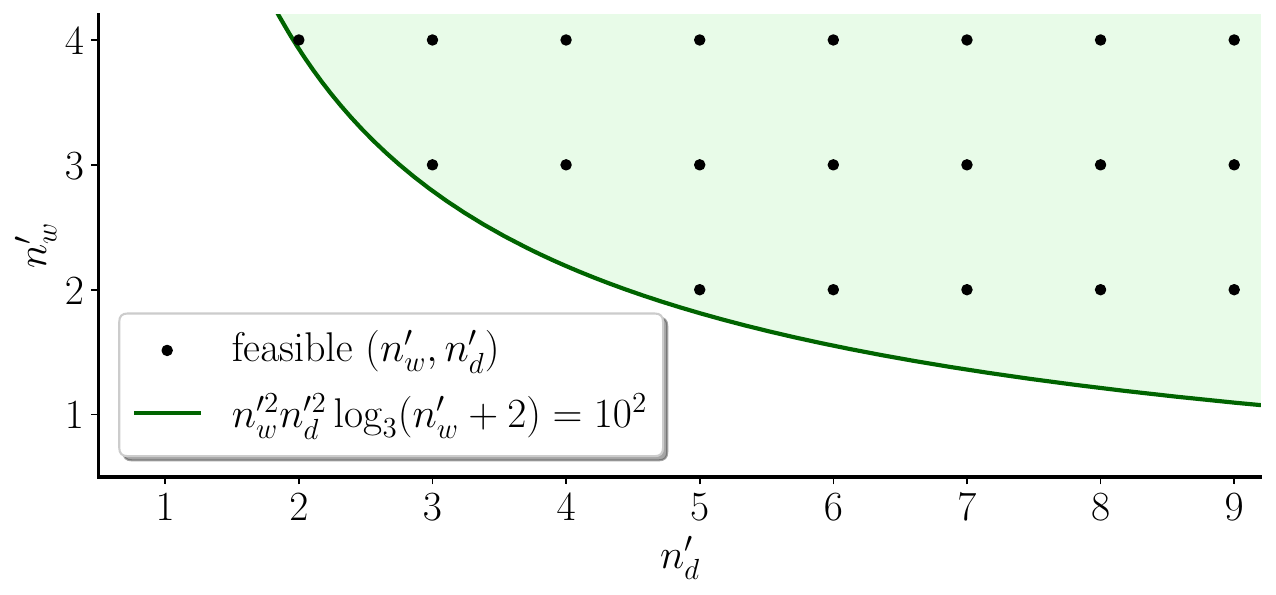}
\vspace{-2em}
\caption{Combinations of $n_w'$ and $n_d'$ satisfying \eqref{equ:main}.}
\label{fig:NL}
\vspace{-1em}
\end{figure}

Theorem \ref{thm:main} characterizes the sufficient complexity of ReLU NN-based MPC policies in terms of the GLC of the MPC policy $u_\text{mpc}(\cdot)$, the error bound $\delta_1$, the constraint set $\mathcal{U}$, and the $\gamma$-sublevel set $\mathcal{X}_\text{inv}$. Based on this result, we draw the following implications for designing ReLU NN-based MPC policies:
\begin{itemize}
\item For a fixed $n_w'$, the smallest $\delta_1$ satisfying \eqref{equ:main} is proportional to $n_d'^{-2/n_x}$. Thus, halving this error bound requires increasing $n_d'$ by a factor of $2^{n_x/2}$. Increasing $n_w'$ gives a slightly faster reduction due to the additional logarithmic factor. However, the number of multiply-accumulate operations for evaluating \eqref{equ:nn} grows quadratically in width and only linearly in depth, making depth more favorable in terms of this cost scaling.
\item The smaller the GLC of the MPC policy $u_\text{mpc}(\cdot)$, the lower the complexity bound in \eqref{equ:main}. Indeed, $L_u$ and $\delta_1$ enter \eqref{equ:main} only through the ratio $L_u/\delta_1$, so halving the GLC has the same effect as doubling the error bound $\delta_1$.
\item The constraint set $\mathcal{U}$ enters the complexity bound in \eqref{equ:main} explicitly only through the ratio $D(\mathcal{U})/d(\mathcal{U})$ of its radius to its inradius, raised to the power $n_x$.
This ratio is invariant to scaling $\mathcal{U}$ and equals one if and only if $\mathcal{U}$ is a ball centered at the origin, so it measures how far $\mathcal{U}$ deviates from such a ball, regardless of how many constraints define $\mathcal{U}$.
\end{itemize}

Unlike the exact representation of piecewise affine control laws \cite{karg2020efficient,fahandezh2020proximity,fabiani2024robust}, whose size grows with the number of critical regions (CRs), the complexity bounds in Theorem \ref{thm:main} do not depend explicitly on the number of CRs.

\begin{remark}\label{rem:existence}Theorem \ref{thm:main} is an a priori result: it guarantees that a network of the stated size satisfying \eqref{tightedapprox} exists, whereas a posteriori verification, e.g., by adapting the mixed-integer programs of \cite{fabiani2022reliably,schwan2023stability,karg2020stability}, checks \eqref{tightedapprox} for a given trained network without knowing beforehand whether its size suffices. Since the stated size is derived from the worst-case point of view, it is generically conservative: a sufficiently large network is guaranteed to work, while smaller networks may also satisfy \eqref{tightedapprox}. Hence, Theorem \ref{thm:main} alone does not indicate whether such a network is preferable to online QP solving in a specific application, which also depends on the computing platform and the implementation and is beyond the scope of this work.
\end{remark}

\section{Network Complexity for the ReLU NN-based MPC Policy in Solving Problem \ref{problem2}}\label{sec:pro2}
In this section, we enhance the closed-loop convergence to a target set through the design of a non-uniform error bound in \eqref{equ:2con1}. To solve Problem \ref{problem2}, we propose a state-aware scaling mechanism that lets the ReLU network meet this bound, and then explicitly quantify the network  complexity. Compared with Theorem \ref{thm:main}, this improves the convergence guarantee and, for small target sets, reduces the complexity bound.

\subsection{Motivation for non-uniform error bounds}\label{sec:motivation}
To ensure convergence of the closed-loop system \eqref{equ:closed_loop} to a target set \(\mathcal{X}_\text{inv}'\), Theorem~\ref{thm:main} requires choosing \(\delta_1\) such that \(\mathcal{V}(c_6 \delta_1^2) \subseteq \mathcal{X}_\text{inv}'\). The smaller \(\mathcal{X}_\text{inv}'\) is, the smaller \(\delta_1\) must be chosen, which, according to \eqref{equ:main}, may lead to an excessively large complexity bound for the ReLU NN-based MPC policy \eqref{equ:nn}.

In contrast to the uniform error bound in \eqref{equ:con1}, larger approximation errors can be tolerated when the system state is relatively far from the origin. If the constant error bound \(\delta_1\) in \eqref{newlyapunova} is replaced with a state-proportional one, i.e., the ReLU NN-based MPC policy \eqref{equ:nn} meets the relative-error bound
\begin{equation}\label{equ:relative_error0}
\| u_\text{nn}(x) - u_\text{mpc}(x) \| \le c_8 \|x\|,
\end{equation}
then
$$
\begin{aligned}
& v_\text{mpc}(Ax + Bu_\text{nn}(x)) \\
& \leq (1 - c_3/c_2) v_\text{mpc}(x) + (c_4 c_8+ c_5 c_8^2)\|x\|^2\\
&\le \left(1 - \left(\frac{c_3}{c_2}-\frac{c_4c_8+c_5c_8^2}{c_1}\right)\right)v_\text{mpc}(x)
\end{aligned}
$$
where the last inequality uses the lower bound in \eqref{equ:exp_stability2}. If $c_8{}>0$ is chosen to satisfy
\begin{equation}\label{equ:stability_condition}
c_4 c_8 + c_5 c_8^2 < {c_1 c_3}/{c_2},
\end{equation}
 the closed-loop system \eqref{equ:closed_loop} exponentially converges to the origin, and hence to any C-set.

Clearly, the approximation error of the ReLU network in \eqref{equ:relative_error0} must decay to zero as $x$ approaches the origin. This is challenging to achieve without using an explicit form of the piecewise affine function $u_\text{mpc}(x)$. For states far from the origin, the admissible error bound under \eqref{equ:relative_error0} could become large, potentially violating the input constraint $u_\text{nn}(x)\in\mathcal{U}$. To resolve these issues and still drive the closed-loop state from $\mathcal{X}_\text{inv}$ into the target set $\mathcal{X}_\text{inv}'$, we bound the admissible error \eqref{equ:relative_error0} from below by a constant near the origin and from above by another constant far from it. This results in the following non-uniform error bound:
\begin{equation}\label{equ:relative_error}
\delta_2(x)=
\begin{dcases}
\underline\delta, & \text{if } \|x\| \leq \underline\delta/c_8,\\
c_8 \|x\|, & \text{if } \underline\delta/c_8 < \|x\| < \bar\delta/c_8,\\
\bar\delta, & \text{if } \|x\| \geq \bar\delta/c_8
\end{dcases}
\end{equation}
where $c_8$ satisfies \eqref{equ:stability_condition}, $\bar\delta$ is the upper bound on the admissible approximation error (cf. Theorem \ref{thm:main}), and $\underline\delta\in(0,\bar\delta]$ plays the role of $\delta_1$ near the origin. The latter is chosen such that $\mathcal{V}(c_6\underline\delta^2)\subseteq\mathcal{X}_\text{inv}'$.
Interestingly, the above non-uniform error bound resembles the error profile of a truncated logarithmic quantizer \cite{you2011attainability}, but contrasts sharply with the uniform bound in \eqref{equ:con1}. To realize such an error bound, we scale both the input and output of the ReLU network, which is analogous to the zoom-in/zoom-out mechanism in quantization \cite{you2011attainability}. Given that a ReLU network is typically far more complex than a quantizer, we need a state-aware scaling mechanism, which is the focus of the next subsection.

\subsection{State-aware scaling for non-uniform error bounds}\label{sec:scaling}

We propose a pair of state-aware input/output scaling functions that enable a static ReLU network to meet non-uniform error bounds.

Our state-aware scaling for the ReLU network is summarized in Fig.~\ref{fig:scaling}. First, we scale the original MPC policy $u_{\text{mpc}}(\cdot)$, where the scaled version $\widetilde{u}_{\text{mpc}}(\cdot)$ acts as the approximation target for the ReLU network. Then, the output of the ReLU network $\widetilde{u}_\text{nn}(\cdot)$ is scaled back, which forms our ReLU NN-based MPC policy. Through the design of state-aware scaling functions $(T, \beta)$, the approximation error between $u_\text{nn}(\cdot)$ and $u_\text{mpc}(\cdot)$ can be bounded in the form of \eqref{equ:relative_error}. We provide details below.

\begin{figure}[tbp]
    \centering
    \includegraphics[width=0.8\linewidth]{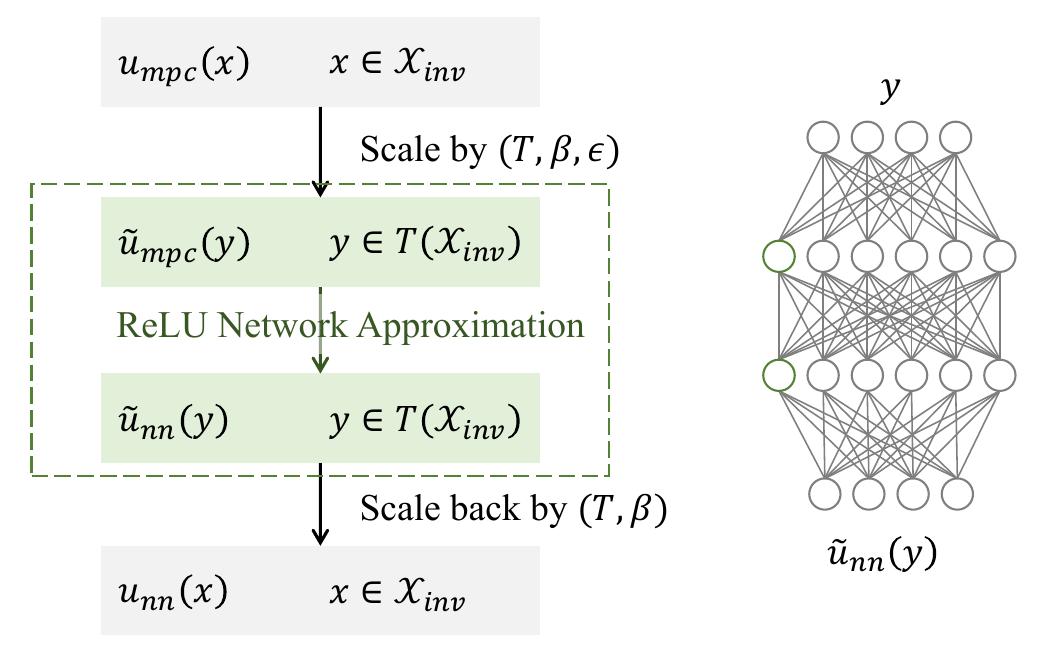}
    \caption{The state-aware scaling approach. The definitions of $\widetilde{u}_\text{mpc}(\cdot)$ and $u_\text{nn}(\cdot)$ are given in \eqref{equ:transformation} and \eqref{equ:trans_recover_first}, respectively.}
    \label{fig:scaling}
\end{figure}

\subsubsection{Scaling back the output of the ReLU network}
To meet non-uniform error bounds in \eqref{equ:relative_error}, we multiply the network output by
\begin{equation*}
\beta(x) = \delta_2(x)/\bar\delta,\quad x\in\mathbb{R}^{n_x}.
\end{equation*}
By \eqref{equ:relative_error}, $\beta(x)$ depends on $x$ only through $\|x\|$ and takes values in $[\underline\delta/\bar\delta,1]$; it is small near the origin and equals one far from it.
The ReLU NN-based MPC policy is then given by
\begin{equation}\label{equ:trans_recover_first}
u_\text{nn}(x) = \widetilde{u}_\text{nn}(T(x)) \beta(x)
\end{equation}
where $\widetilde{u}_\text{nn}(\cdot)$ denotes the ReLU network and $T(\cdot): \mathbb{R}^{n_x}\rightarrow \mathbb{R}^{n_x}$ is an invertible function to scale the input of the ReLU network, designed below. For $y=T(x)$, the approximation error of \eqref{equ:trans_recover_first} is
\begin{equation}\label{allowerror}
\|u_\text{nn}(x)-u_\text{mpc}(x)\|=\beta(x)\,\|\widetilde{u}_\text{nn}(y)-u_\text{mpc}(x)/\beta(x)\|.
\end{equation}
Hence, if the network approximates the scaled policy $u_\text{mpc}(x)/\beta(x)$ with the uniform error bound $\bar\delta$, then $\|u_\text{nn}(x)-u_\text{mpc}(x)\|\le\beta(x)\bar\delta=\delta_2(x)$, which is \eqref{equ:2con1}. That is, the output scaling turns a uniform error bound of the network into the non-uniform error bound $\delta_2(\cdot)$ of the policy.

\begin{figure}[tbp]
     \centering
     \includegraphics[width=1\linewidth]{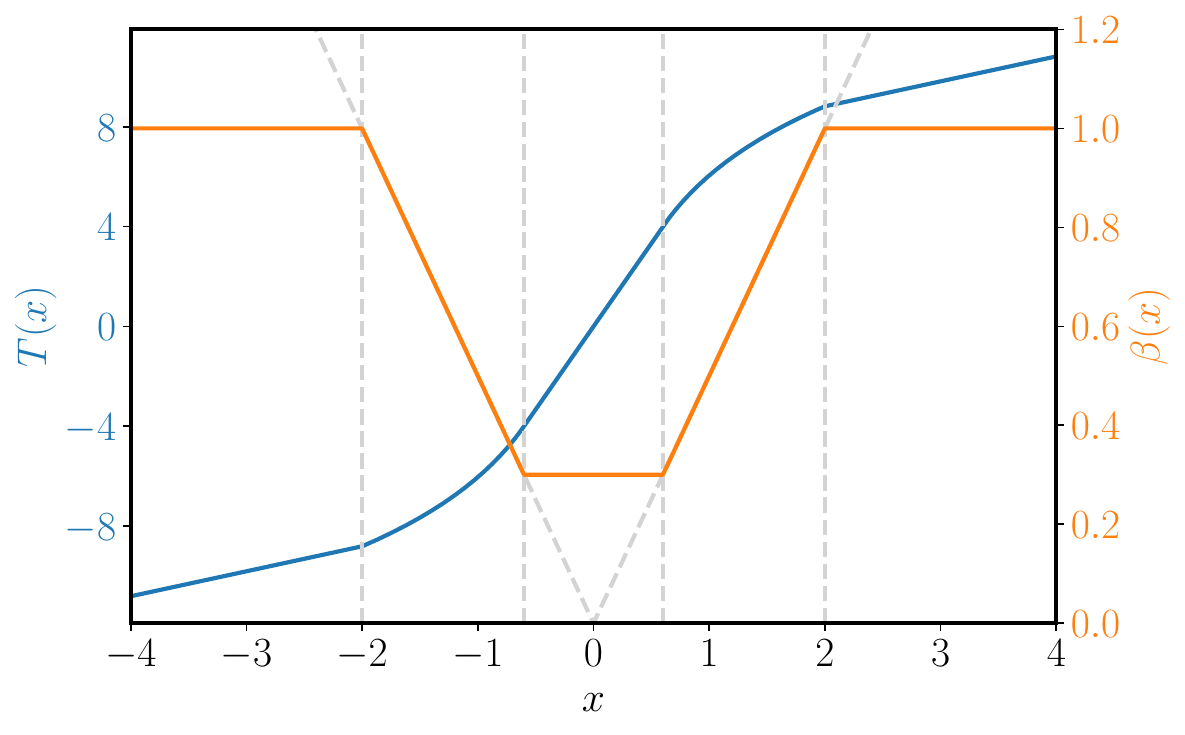}
     \caption{An example of $\beta(x)=\delta_2(x)/\bar\delta$ and the input scaling function $T(x)$.}
     \label{fig:double_yaxis}
 \end{figure}

\begin{figure*}[t]
     \centering
     \includegraphics[width=\textwidth]{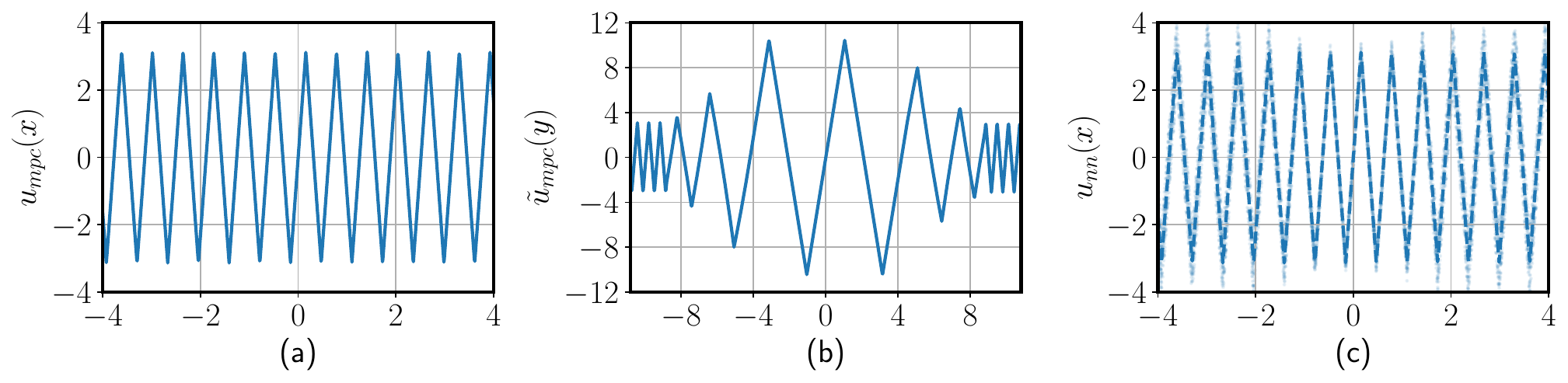}
     \caption{An example of $u_\text{mpc}(\cdot)$, $\widetilde{u}_\text{mpc}(\cdot)$ and $u_\text{nn}(\cdot)$. The shaded area in subfigure (c)  represents the admissible error bound $\delta_2(x)$.}
     \label{fig:combined_plot}
 \end{figure*}

\subsubsection{Scaling the input of the ReLU network} Since $\beta(x)$ is small near the origin, the scaled policy $u_\text{mpc}(x)/\beta(x)$ is steep there, e.g., it equals $(\bar\delta/\underline\delta)u_\text{mpc}(x)$ on $\mathcal{B}(\underline\delta/c_8)$, which can increase the sufficient network complexity bound for the same uniform approximation error \cite{shen2019deep}. The input scaling $T(\cdot)$ compensates for this by stretching the input space near the origin: stretching a neighborhood of the origin by the factor $1/\beta$ flattens the target, viewed as a function of $y=T(x)$, by the same factor. We therefore let the radial stretching rate of $T(\cdot)$ be inversely proportional to $\beta$:
\begin{equation}\label{equ:Tdef}
T(x)=x\int_0^{1} \dfrac{1 + \mathbbm{1}( s\|x\| \le\bar\delta/c_8)}{\beta(sx)}\mathrm{d}s.
\end{equation}
That is, $\mathrm{d}\|T(x)\|/\mathrm{d}\|x\|=(1+\mathbbm{1}(\|x\|\le\bar\delta/c_8))/\beta(x)$ for $\|x\|\neq\bar\delta/c_8$, i.e., $T(\cdot)$ stretches the radius at the rate $2/\beta$ on $\mathcal{B}(\bar\delta/c_8)$ and at the rate one outside it, where $\beta=1$. The factor two accounts for the variation of $\beta$, so that the GLC of the target, as a function of $y$, does not exceed that of $u_\text{mpc}(\cdot)$ (cf. Lemma \ref{lem:3Lip6}). In fact, the map $T(\cdot)$ admits an explicit piecewise smooth expression, illustrated in Fig. \ref{fig:double_yaxis}:
\begin{equation}\label{equ:Texplicit}
T(x) =
\begin{dcases}
    \frac{2\bar\delta}{\underline\delta}x,
    & \text{if } \|x\| \le \tfrac{\underline\delta}{c_8}, \\
    \dfrac{2\bar\delta}{c_8}\left(1+\ln\left(\dfrac{c_8\|x\|}{\underline\delta}\right)\right) \dfrac{x}{\|x\|},
    & \text{if } \tfrac{\underline\delta}{c_8} < \|x\| < \tfrac{\bar\delta}{c_8}, \\
    x + \dfrac{\bar\delta}{c_8}\left(1+2 \ln \left(\dfrac{\bar\delta}{\underline\delta}\right)\right) \dfrac{x}{\|x\|},
    & \text{if } \|x\| \ge \tfrac{\bar\delta}{c_8}.
\end{dcases}
\end{equation}

To ensure the input constraint $u_\text{nn}(x)\in\mathcal{U}$ as in Section \ref{sec:projection}, we retract $u_\text{mpc}(x)$ into the tightened set $\mathcal{U}'(\epsilon\beta(x))$ of \eqref{tightset} by the factor $1-\epsilon\beta(x)/d(\mathcal{U})$, where $\epsilon$ is the error bound of the network specified in Lemma \ref{lem:transformation}.

Overall, the following scaled and retracted value of the MPC policy acts as the approximation target of the ReLU network:
 \begin{align}\label{equ:transformation}
 \widetilde{u}_\text{mpc}(y) &= \frac{1}{\beta(x)}\left(1-\frac{\epsilon\beta(x)}{d(\mathcal{U})}\right)u_\text{mpc}(x),~\text{and}~ y=T(x).
\end{align}
For $\epsilon=0$, it reduces to the scaled policy in \eqref{allowerror}. Since $\mathcal{B}(d(\mathcal{U}))\subseteq\mathcal{U}$, $\mathcal{U}$ is convex and $\epsilon\beta(x)<d(\mathcal{U})$ (cf. Lemma \ref{lem:transformation}), the retracted value $(1-\epsilon\beta(x)/d(\mathcal{U}))u_\text{mpc}(x)$ indeed lies in $\mathcal{U}'(\epsilon\beta(x))$.

Under the above design, we have the following results.
\begin{lemma}\label{lem:transformation}
Let $0\le{}\epsilon \le \bar\delta d(\mathcal{U})/(2D(\mathcal{U}))$, so that $\epsilon\beta(x)\le\epsilon<d(\mathcal{U})$ for every $x$ since $\bar\delta\le d(\mathcal{U})$. If $$\|\widetilde{u}_\text{nn}(y)-\widetilde{u}_\text{mpc}(y)\|\leq \epsilon, ~\forall y \in T(\mathcal{X}_\text{inv}),$$ then the ReLU NN-based MPC policy in \eqref{equ:trans_recover_first} satisfies, for every $x\in\mathcal{X}_\text{inv}$,
\begin{equation}
\|u_\text{nn}(x)-u_\text{mpc}(x)\|\leq \delta_2(x) \quad \text{and} \quad u_\text{nn}(x)\in \mathcal{U}.
\end{equation}
\end{lemma}
\begin{proof}
See Appendix \ref{sec:proof_transformation}.
\end{proof}

As $\delta_1\leq \bar\delta$, the admissible error bound of the ReLU network in Lemma \ref{lem:transformation} is at least as large as that of \eqref{tightedapprox}, and independent of the target set \(\mathcal{X}_\text{inv}'\).

\begin{lemma}\label{lem:3Lip6}
The input scaling function $T(\cdot)$ is invertible and, for $\epsilon$ as in Lemma \ref{lem:transformation}, the GLC of
$\widetilde{u}_\text{mpc}(\cdot)$ over $T(\mathcal{X}_\text{inv})$ satisfies
\begin{equation}\label{equ:3Lip8}
\mathcal{L}(\widetilde{u}_\text{mpc},T(\mathcal{X}_\text{inv})) \le L_u.
\end{equation}
\end{lemma}
\begin{proof}
See Appendix \ref{sec:proof_3Lip6}.
\end{proof}

\subsubsection{An illustrative example}
To illustrate the proposed scaling approach, consider a one-dimensional example where $u_\text{mpc}(x)=2\arcsin(\sin(10x))$ is a continuous piecewise linear triangle wave with period $\pi/5$ and amplitude $\pi$, and satisfies $u_\text{mpc}(0)=0$. We set the parameters to $\underline\delta=0.3$, $\bar\delta=1$, $\epsilon=0$, and $c_8=0.5$, and consider $x\in[-4,4]$. Here $\widetilde u_\text{nn}(\cdot)$ is emulated by adding to $\widetilde u_\text{mpc}(\cdot)$ an error drawn uniformly from $[-\bar\delta,\bar\delta]$, i.e., an approximation with the uniform error bound $\bar\delta$. Fig.~\ref{fig:double_yaxis} visualizes the scaling function $\beta(x)$ and the input scaling function $T(x)$. The effect is shown in Fig. \ref{fig:combined_plot}: near the origin, $u_\text{nn}(x)$ tracks $u_\text{mpc}(x)$ with high precision, whereas for states far from the origin, a larger deviation is permitted. This illustrates how the scaling concentrates approximation accuracy where it is most critical for stability.

\subsection{Network complexity bounds for solving Problem \ref{problem2}}

Finally, the following theorem gives explicit complexity bounds under which the NN-based MPC policy \eqref{equ:trans_recover_first} solves Problem \ref{problem2}.

\begin{theorem}[Network Complexity for Problem \ref{problem2}]\label{thm:scaling}\label{thm:stability_relative_error}\label{thm:main2}
Consider Problem \ref{problem2} with any specified C-set $\mathcal{X}_\text{inv}'\subseteq\mathcal{X}_\text{inv}$. If there is a pair of positive integers $(n_w',n_d')$ satisfying
\begin{equation}\label{equ:2main}
\begin{aligned}
&n_w'^2n_d'^2\log_3(n_w'+2)\\[1mm]
&\geq~\left(\frac{524\sqrt{n_x n_u}\,D(\mathcal{U})\,D_2\,{L}_u}{\bar{\delta}~d(\mathcal{U})}\right)^{n_x},
\end{aligned}
\end{equation}
there exists a ReLU NN-based MPC policy of the form \eqref{equ:trans_recover_first} whose ReLU network has width and depth
\begin{subequations}
\begin{align}
n_w&=n_u3^{n_x+3}\max\{n_x\lfloor n_w'^{1/n_x}\rfloor,n_w'+2\},\label{equ:2width}\\
n_d&=11n_d'+19+2n_x,\label{equ:2depth}
\end{align}
\end{subequations}
and which solves Problem \ref{problem2} with $\delta_2$ in \eqref{equ:relative_error}. Here, $c_0,c_1,\ldots,c_6$ are as in Theorem \ref{thm:main}, $c_8$ satisfies \eqref{equ:stability_condition}, $\underline\delta = \min\{\bar\delta, \sqrt{c_1/c_6}\cdot d(\mathcal{X}_\text{inv}')\}$, and $D_2= D(\mathcal{X}_\text{inv}) + ({\bar\delta}/{c_8}) \left(1+2\ln\left({\bar\delta}/{\underline\delta}\right)\right)$.
\end{theorem}

\begin{proof}
See Appendix \ref{sec:proof_thm2}.
\end{proof}

Remark \ref{rem:existence} applies to Theorem \ref{thm:main2} as well, with \eqref{tightedapprox} replaced by the hypothesis of Lemma \ref{lem:transformation}.

\begin{remark}For small target sets, Theorem \ref{thm:main2} establishes a smaller complexity bound than Theorem \ref{thm:main}. To drive the state into the same target set $\mathcal{X}_\text{inv}'$, we take in Theorem \ref{thm:main} the largest $\delta_1\le\bar\delta$ for which \eqref{equ:exp_stability2} ensures $\mathcal{V}(c_6\delta_1^2)\subseteq\mathcal{X}_\text{inv}'$, namely $\delta_1=\underline\delta$. Then \eqref{equ:main} and \eqref{equ:2main} differ only in the factor
\begin{equation}\label{equ:gain}
\frac{D(\mathcal{X}_\text{inv})}{\underline\delta}\quad\text{versus}\quad\frac{D_2}{\bar\delta}=\frac{D(\mathcal{X}_\text{inv})}{\bar\delta}+\frac{1+2\ln(\bar\delta/\underline\delta)}{c_8},
\end{equation}
which grows as $1/\underline\delta$ in \eqref{equ:main} but only as $\ln(1/\underline\delta)$ in \eqref{equ:2main}. As $\mathcal{X}_\text{inv}'$ shrinks, $\underline\delta$ becomes proportional to $d(\mathcal{X}_\text{inv}')$, so the bound of Theorem \ref{thm:main2} grows only polylogarithmically, instead of polynomially, in $1/d(\mathcal{X}_\text{inv}')$.
\end{remark}

\section{Numerical Examples}\label{sec:simu}
This section illustrates the two main results on two systems. The double integrator has an explicit MPC solution, so that the constants of Theorem \ref{thm:main} can be computed exactly and the approximation error of a trained network can be certified. On this system we follow Theorem \ref{thm:main} from the constants to the closed loop (Section \ref{sec:small}). For the oscillating masses system, we examine numerically the state-aware scaling of Theorem \ref{thm:main2} (Section \ref{sec:big}).
\subsection{Systems and constants}\label{sec:systems}
\textbf{Double integrator.} We consider the discrete-time model
\begin{equation}\label{equ:small}
A=\begin{bmatrix}1&1\\0&1\end{bmatrix},\quad
      B=\begin{bmatrix}0.5\\1\end{bmatrix},
\end{equation}
with the MPC \eqref{equ:mpc} parameterized by $Q=I_2$, $R=0.1$, $n_p=10$, $\mathcal{X}=\{x:|x_1|\le10,\,|x_2|\le5\}$ and $\mathcal{U}=[-1,1]$. The matrix $P$ is the solution of the discrete algebraic Riccati equation \cite{kuvcera1972discrete}, $\mathcal{X}_f$ is the maximal positively invariant set of the LQR closed loop, and $\mathcal{X}_\text{inv}$ is the sublevel set \eqref{sublevel} with $\gamma=0.99\,\gamma_{\max}$, which covers $74\%$ of the area of the feasible set $\mathcal{F}$. The explicit MPC solution \cite{borrelli2017predictive} is piecewise affine with $115$ CRs, from which the constants of Theorem \ref{thm:main} in Table \ref{tab:const} are computed. Here $L_u$ is the largest gain over the CRs meeting $\mathcal{X}_\text{inv}$, $c_1,c_2,c_3$ are global optima of two-variable quadratic programs over these regions, and $c_0$ is the supremum of $\|\nabla v_\text{mpc}(x)\|/\|x\|$ over the $\|B\|\bar\delta$-neighborhood of $\mathcal{X}_\text{inv}$. On this neighborhood $v_\text{mpc}$ is convex, piecewise quadratic and continuously differentiable. Hence \eqref{equ:lemma2-part} holds with $c_4$ the largest ratio $|B^\top\nabla v_\text{mpc}(Ax+Bu_\text{mpc}(x))|/\|x\|$ over $\mathcal{X}_\text{inv}$ and $c_5$ the largest curvature $B^\top\nabla^2 v_\text{mpc}B/2$, both computed region by region, and the resulting $c_6$ is $19$ times smaller than with the closed forms after \eqref{equ:lemma2-part}. The constants $c_6$, $c_7$ and $\bar\delta$ follow in closed form.

\begin{table}[tbp]
\centering
\caption{Constants of Theorem \ref{thm:main} for the system \eqref{equ:small}, with $\gamma=0.99\gamma_{\max}$.}
\label{tab:const}
\vspace{-0.5em}
\footnotesize
\setlength{\tabcolsep}{4pt}
\renewcommand{\arraystretch}{1.1}
\begin{tabular*}{\columnwidth}{@{\extracolsep{\fill}\hspace{4pt}}p{26pt}p{40pt}|p{36pt}p{30pt}|p{14pt}p{52pt}@{\hspace{4pt}}}
\toprule
$\|A\|$ & $1.618$ & $\gamma$ & $185.9$ & $c_3$ & $1.000$ \\
$\|B\|$ & $1.118$ & $D(\mathcal{X}_\text{inv})$ & $10.95$ & $c_4$ & $56.84$ \\
$d(\mathcal{U})$ & $1$ & $d(\mathcal{X}_\text{inv})$ & $3.426$ & $c_5$ & $246.9$ \\
$D(\mathcal{U})$ & $1$ & $c_0$ & $70.39$ & $c_6$ & $7.652\times10^{5}$ \\
$L_u$ & $1.412$ & $c_1$ & $1.069$ & $c_7$ & $1.143\times10^{-3}$ \\
$\gamma_{\max}$ & $187.7$ & $c_2$ & $15.83$ & $\bar\delta$ & $7.793\times10^{-3}$ \\
\bottomrule
\end{tabular*}
\end{table}

\textbf{Oscillating masses.} We consider the oscillating masses system \cite{wang2009fast} with the damping and the actuators of \cite{chen2022large}, sketched in Fig. \ref{fig:sys}, with $n_x=12$ and $n_u=3$. Its continuous-time model
\begin{equation*}
\dot x=\begin{bmatrix}0&I_6\\-2I_6+S&-2I_6+0.1S\end{bmatrix}x+\begin{bmatrix}0\\F\end{bmatrix}u,
\end{equation*}
where $S$ has ones on the first sub- and superdiagonals and zeros elsewhere and $F=[e_1,-e_1,e_2,e_3,-e_2,e_3]^\top$ with the unit vectors $e_i\in\mathbb{R}^3$, is discretized by a first-order hold with period $0.5$. The MPC \eqref{equ:mpc} is parameterized by $Q=I_{12}$, $R=I_3$, $n_p=20$, $\mathcal{X}=\{x:|x_i|\le4\}$, $\mathcal{U}=\{u:|u_i|\le0.5\}$ and $\mathcal{X}_f=\mathcal{X}$. The matrix $P$ is the solution of the discrete algebraic Riccati equation, and $\mathcal{X}_\text{inv}$ is the sublevel set \eqref{sublevel} with $\gamma=20$. Since the explicit MPC solution is impractical to compute for this system, we take $c_1=\lambda_{\min}(P)$ and $c_3=\lambda_{\min}(Q)$, confirmed by the decrease of $v_\text{mpc}$ on sampled states, and estimate $c_0$, $c_2$, $L_u$ and the radii of $\mathcal{X}_\text{inv}$ on these states. This gives $\bar\delta=5.1\times10^{-2}$, $L_u=0.97$, $D(\mathcal{X}_\text{inv})=3.9$ and $d(\mathcal{X}_\text{inv})=1.9$. We take $c_8=3.8\times10^{-2}$, which is $0.9$ times the supremum of the values allowed by \eqref{equ:stability_condition}. With these constants the RHS of \eqref{equ:main} is of the order of $10^{67}$ for $\delta_1=\bar\delta$.

\begin{figure}[tbp]
    \centering
    \includegraphics[width=1\linewidth]{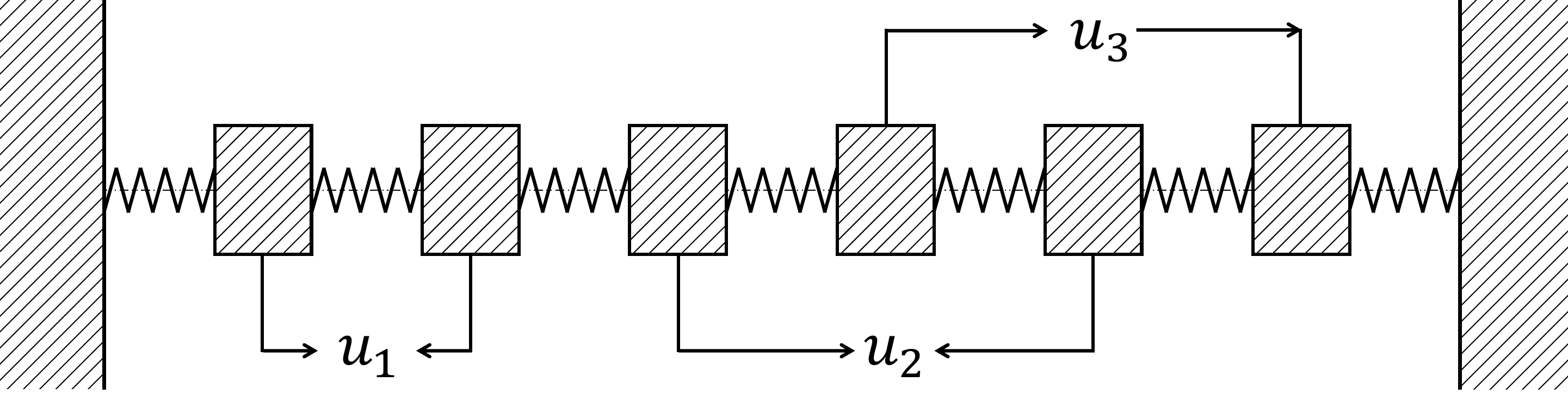}
    \vspace{-2em}
    \caption{Sketch of the oscillating masses system \cite{wang2009fast,chen2022large}, with the damping not shown.}
    \label{fig:sys}
\end{figure}

\subsection{Guarantee under a uniform error bound}\label{sec:small}
\textbf{Network size prescribed by Theorem \ref{thm:main}.} We start from the constants and ask how large a network the theorem prescribes. Table \ref{tab:const} gives $\bar\delta=7.79\times10^{-3}$, so that Lemma \ref{lem:projection} admits a tolerance of $\epsilon=\bar\delta/2=3.90\times10^{-3}$ on the projected target for $\delta_1=\bar\delta$. With this $\delta_1$, \eqref{equ:main} reads $n_w'^2n_d'^2\log_3(n_w'+2)\ge2.16\times10^{12}$. For instance, with $n_d'=200$ the smallest admissible $n_w'$ is $2740$, which gives $n_w=6.66\times10^{5}$ and $n_d=2223$ through \eqref{equ:width} and \eqref{equ:depth}, or about $9.9\times10^{14}$ parameters. This size suffices but is conservative, as Remark \ref{rem:existence} anticipates.

\textbf{Training and verifying a network.} A trained network instead inherits the guarantee once \eqref{tightedapprox} has been verified for it. We therefore train a ReLU network with three hidden layers of width $8$, that is, $177$ parameters, to approximate the projected policy $\Pi_{\mathcal{U}'(\epsilon)}(u_\text{mpc}(\cdot))$ of Lemma \ref{lem:projection}. We take $\epsilon=8.91\times10^{-4}$, below the admissible $\bar\delta/2$, so that $\delta_1$ and hence the guaranteed set are smaller. Since \eqref{tightedapprox} bounds the maximum error, we minimize a smooth maximum of the pointwise errors. Since $u_\text{nn}$ is piecewise affine and so is the projected policy on each CR, upper bounds $\hat\delta$ on $\max_{x\in\mathcal{X}_\text{inv}}\|u_\text{nn}(x)-\Pi_{\mathcal{U}'(\epsilon)}(u_\text{mpc}(x))\|$ and on $\max|u_\text{nn}|$ follow from mixed-integer linear programs over a polyhedral outer approximation of $\mathcal{X}_\text{inv}$, solved to global optimality. We obtain $\hat\delta=1.43\times10^{-5}\le\epsilon$ and $\max|u_\text{nn}|\le0.9991<1$. Since $d(\mathcal{U})=D(\mathcal{U})$, the proof of Lemma \ref{lem:projection} gives $\delta_1=\hat\delta+\epsilon=9.05\times10^{-4}\le\bar\delta$, so that Lemma \ref{lem:convergence} applies: the closed loop converges to $\mathcal{V}(c_6\delta_1^2)=\mathcal{V}(0.63)$.

\textbf{Closed-loop behavior of the network.} Fig. \ref{fig:small_closed_loop} shows $60$ closed-loop trajectories of this network, started at $0.97$ of the boundary of $\mathcal{X}_\text{inv}$ and run for $120$ steps. The MPC input is saturated on $72\%$ of $\mathcal{X}_\text{inv}$ (grey), and the network respects the input constraint there, as verified above. No state leaves $\mathcal{X}_\text{inv}$, and every trajectory enters the guaranteed set $\mathcal{V}(0.63)$ within $9$ steps. The trajectories then continue to the origin, beyond what Theorem \ref{thm:main} guarantees for $\delta_1>0$, since Lemma \ref{lem:convergence} allows the error $\delta_1$ at every state.

\begin{figure}[tbp]
\centering
\includegraphics[width=\columnwidth]{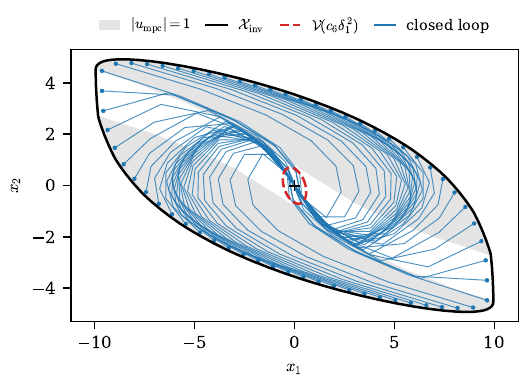}
\caption{Closed loop of the certified network for the system \eqref{equ:small}: $60$ trajectories (blue) started at $0.97$ of the boundary of $\mathcal{X}_\text{inv}$ (black), the set $\mathcal{V}(c_6\delta_1^2)$ of Lemma \ref{lem:convergence} (red dashed), and the region on which the MPC input is saturated (grey).}
\label{fig:small_closed_loop}
\end{figure}

\textbf{Convergence rate and guaranteed set.} Lemma \ref{lem:convergence} guarantees convergence to $\mathcal{V}(c_6\delta_1^2)$ at the rate $v_\text{mpc}(x_k)-c_6\delta_1^2\le q^k\gamma$ with $q=0.968$. Fig. \ref{fig:small_decay} shows that the trajectories respect this bound, and outside the guaranteed set $v_\text{mpc}-c_6\delta_1^2$ contracts per step by a factor of at most $0.932$, close to $q$. The guaranteed set $\mathcal{V}(0.63)$ has inner and outer radii $0.51$ and $0.77$, against $3.43$ and $10.95$ for $\mathcal{X}_\text{inv}$, and covers $1.0\%$ of its area. For comparison, the GLC-based bound \eqref{traditional_result} of \cite{fabiani2022reliably}, with $L_v=273$ over $\mathcal{X}_\text{inv}$, gives the set $\mathcal{V}(4.37)$ for this $\delta_1$.

\begin{figure}[tbp]
\centering
\includegraphics[width=\columnwidth]{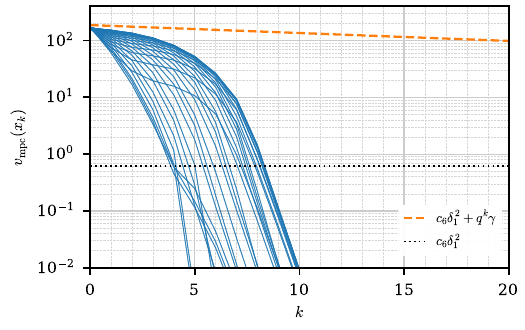}
\vspace{-0.5em}
\caption{Decay of $v_\text{mpc}(x_k)$ along the $60$ trajectories of Fig. \ref{fig:small_closed_loop} (blue), the bound $c_6\delta_1^2+q^k\gamma$ of Lemma \ref{lem:convergence} (dashed) and the level $c_6\delta_1^2$ (dotted).}
\label{fig:small_decay}
\end{figure}

\subsection{Uniform versus non-uniform error bounds}\label{sec:big}

On the oscillating masses system, we compare two ReLU networks for the target sets $\mathcal{X}_\text{inv}'=\mathcal{B}(d)$ with $d=0.1$ and $d=0.001$: the uniform network, trained for the uniform error bound of Theorem \ref{thm:main}, and the scaled network, which uses the state-aware scaling of Theorem \ref{thm:main2}. Both networks have three hidden layers of width $192$, that is, $77187$ parameters. For these target sets the RHS of \eqref{equ:main} at $\delta_1=\underline\delta$ is of the order of $10^{80}$ for $d=0.1$ and $10^{104}$ for $d=0.001$, against $10^{73}$ and $10^{77}$ for \eqref{equ:2main}. Each network is trained in the coordinates in which its theorem asks for a uniform error bound (cf. \eqref{tightedapprox} and Lemma \ref{lem:transformation}): the uniform network on $2\times10^{6}$ states sampled uniformly over $\mathcal{X}_\text{inv}$, and the scaled network on $1.6\times10^{6}$ of these states and on $0.4\times10^{6}$ states sampled uniformly in the scaled coordinates $y=T(x)$. The errors below are measured on $3\times10^{5}$ independent test states.

\textbf{Approximation accuracy.} Fig. \ref{fig:vc_error} shows the errors of the two networks against $\|x\|$, together with the constant bound $\underline\delta$ of Theorem \ref{thm:main} and the bound $\delta_2(x)$ of Theorem \ref{thm:main2}. The error of the scaled network follows the ramp of $\delta_2(x)$ and lies below it at $99.9\%$ and $99.4\%$ of the test states for $d=0.1$ and $0.001$. It is small near the origin because the output is scaled down there by $\beta(x)=\underline\delta/\bar\delta$, which is $0.084$ and $0.00084$, so that $\|u_\text{nn}(0)\|=8.2\times10^{-6}$ and $3.5\times10^{-7}$. The error of the uniform network is almost flat, in line with the uniform bound it is trained for, about $10^{-3}$ near the origin for both values of $d$, with $\|u_\text{nn}(0)\|=8.3\times10^{-4}$ and $1.1\times10^{-3}$. It therefore lies below $\underline\delta$ at $99.9\%$ of the test states for $d=0.1$, but at only $0.02\%$ for $d=0.001$. The price of the scaling is a larger error away from the origin, with largest values of $0.09$ and $0.18$ against $0.011$ and $0.008$.

\begin{figure}[tbp]
\centering
\includegraphics[width=\columnwidth]{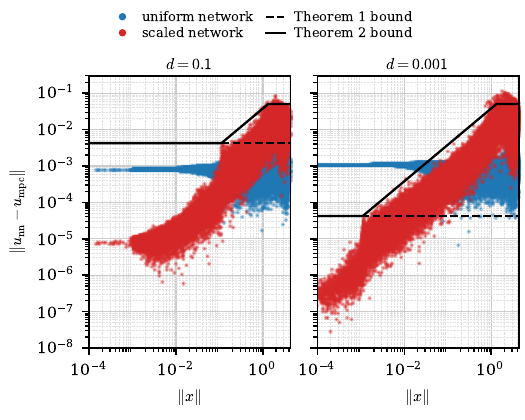}
\vspace{-0.5em}
\caption{Error $\|u_\text{nn}(x)-u_\text{mpc}(x)\|$ of the two networks on the test states against $\|x\|$, for $d=0.1$ and $d=0.001$. The dashed line is the constant bound $\underline\delta$ of Theorem \ref{thm:main}, the solid line the bound $\delta_2(x)$ of Theorem \ref{thm:main2}.}
\label{fig:vc_error}
\end{figure}

\textbf{Closed-loop behavior.} Fig. \ref{fig:vc_decay} shows $v_\text{mpc}(x_k)$ along closed loops of the two networks started at $40$ random states at $0.97$ of the boundary of $\mathcal{X}_\text{inv}$, with the input saturated to $\mathcal{U}$. No state leaves $\mathcal{X}_\text{inv}$, and all trajectories come to rest at an equilibrium whose distance from the origin is of the order of $\|u_\text{nn}(0)\|$: $8\times10^{-4}$ and $9\times10^{-4}$ for the uniform network, and $8.1\times10^{-6}$ and $3.5\times10^{-7}$ for the scaled network, for $d=0.1$ and $0.001$. As a result, for $d=0.001$ the trajectories of the scaled network enter $\mathcal{V}(c_1d^2)=\mathcal{V}(c_6\underline\delta^2)$, the invariant set in the proof of Theorem \ref{thm:main2}, within $44$ steps, whereas those of the uniform network stop above the level $c_1d^2$. Started near the origin, at $\|x_0\|=10^{-3}$, the uniform network accumulates about $340$ and $560$ times the cost of the MPC over $1000$ steps because it comes to rest away from the origin, whereas the cost of the scaled network stays within $4\%$ of that of the MPC. Started near the boundary of $\mathcal{X}_\text{inv}$, both networks stay close to the MPC cost, with relative excess costs of $5\times10^{-4}$ to $6\times10^{-4}$ for the scaled network and about $9\times10^{-5}$ for the uniform network.

\begin{figure}[tbp]
\centering
\includegraphics[width=\columnwidth]{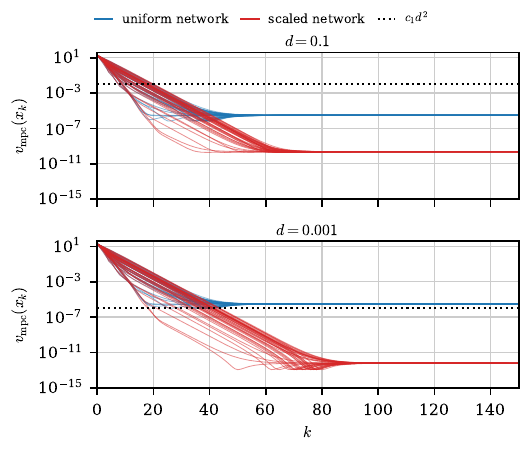}
\vspace{-0.5em}
\caption{Closed loops of the two networks from $40$ random states at $0.97$ of the boundary of $\mathcal{X}_\text{inv}$, with the input saturated to $\mathcal{U}$: $v_\text{mpc}(x_k)$ against $k$ for $d=0.1$ and $d=0.001$. The dotted line is the level $c_1d^2$, whose sublevel set lies in the target set $\mathcal{B}(d)$.}
\label{fig:vc_decay}
\end{figure}

\subsection{From the guarantees to trained networks}\label{sec:gap}

The theorems give a network size that suffices and the error conditions that such a network must meet. Obtaining a network that meets these conditions is in general a further step of training and verification. The theory does, however, state explicitly what the training must achieve, and this suggests three approaches for future work. The first is to train on the condition itself: minimize the largest ratio of $\|u_\text{nn}(x)-u_\text{mpc}(x)\|$ to $\delta_2(x)$ in a smooth-maximum form, and add the states where it is largest. The second is to build the structure of the law into the network: since $u_\text{mpc}$ is the LQR law on a ball around the origin, a linear term leaves the network to approximate a remainder that vanishes there, and for a box $\mathcal{U}$ a saturation layer keeps the output in $\mathcal{U}$ without increasing the error, because $u_\text{mpc}(x)\in\mathcal{U}$ and the projection is non-expansive. The third is to train on the decrease condition \eqref{newlyapunova}, which is all that Lemma \ref{lem:convergence} uses and is weaker than the error bound.

\section{Conclusion and Future Work}\label{sec:conclusion}
This paper has focused on the open problem of relating the network complexity of ReLU NN-based MPC policies to closed-loop performance. We have established explicit upper bounds on the width and depth of ReLU NNs to guarantee constraint satisfaction and stability. These bounds have been expressed as functions of the desired approximation error, MPC problem parameters, the Lipschitz constant of the control law, and the constraint set. Furthermore, we have proposed a state-aware scaling approximation method that has been shown to reduce NN complexity for small target sets and has made explicit the trade-off between network complexity and the size of the target set. Along the way, we have obtained two results that may be of independent interest: (a) the state-dependent Lipschitz continuity of the cost function, and (b) bounds for projection onto tightened sets. Future work will sharpen the complexity bounds by better exploiting the specific structure of the MPC policy, extend the results to nonlinear, tracking, and time-varying MPC with state/output constraints, and study practical methods for the training and verification of networks that attain the accuracy required by the theorems.

\appendix
In this section, we give the proofs for Sections \ref{sec:main_results} and \ref{sec:pro2}, together with the auxiliary Lemmas \ref{lem:3Lip2}--\ref{lem:3Lip4}. Appendix \ref{sec:nn} restates the approximation result used in both theorems.

\subsection{Proof of Results in Section \ref{sec:main_results}}\label{sec:appA}
\subsubsection{Proof of Lemma \ref{lem:3Lip}}\label{sec:proof_3Lip}
The feasible set $\mathcal{F}\supseteq\mathcal{X}_\text{inv}$ of \eqref{equ:mpc} is convex and compact because $\mathcal{X}$, $\mathcal{U}$ and $\mathcal{X}_f$ are C-sets. Let $u_{i|t}^*(x)$ denote the control input at prediction step $i$ of the optimal solution of \eqref{equ:mpc}, and define the optimal predicted state trajectory as
\begin{equation*}
x_{i|t}^*(x) = A^ix + \sum_{j=0}^{i-1} A^{i-j-1} B u_{j|t}^*(x).
\end{equation*}
Since $\mathcal{X}$, $\mathcal{U}$, $\mathcal{X}_f$ are polytopic, \eqref{equ:mpc} is a multi-parametric quadratic program. Since $R\in\mathbb{S}^{n_u}_{++}$, its optimal solution is unique and continuous piecewise affine in $x$ on $\mathcal{F}$ \cite[Corollary 5.2 and Remark 5.8]{borrelli2017predictive}. Hence, on the $\ell$-th of the finitely many CRs of this explicit MPC solution, $u^*_{i|t}(\cdot)$ is affine with some gain $K_{i,\ell}$. By the display above, $x^*_{i|t}(\cdot)$ is affine there with gain $\Psi_{i,\ell}=A^i+\sum_{j<i}A^{i-j-1}BK_{j,\ell}$. Since $\mathcal{F}$ is convex, a continuous piecewise affine map on $\mathcal{F}$ is Lipschitz continuous with its largest gain norm as a Lipschitz constant. Therefore, $u^*_{i|t}(\cdot)$ and $x^*_{i|t}(\cdot)$ are Lipschitz continuous on $\mathcal{F}$ with constants
\begin{equation}\label{equ:Lxu}
L_{u,i}=\max_\ell\|K_{i,\ell}\|,\quad L_{x,i}=\max_\ell\|\Psi_{i,\ell}\|.
\end{equation}
Since $\mathcal{X}_\text{inv}$ is convex and $u_\text{mpc}(\cdot)=u^*_{0|t}(\cdot)$, the same argument gives $L_u\le\max_\ell\|K_{0,\ell}\|$ over the CRs meeting $\mathcal{X}_\text{inv}$. This bound can replace $L_u$ throughout, as $L_u$ serves only as a Lipschitz constant.

At $x=0$, the zero input sequence is feasible with zero cost. It is optimal because the cost is non-negative, so $v_\text{mpc}(0)=0$. By uniqueness, $u^*_{i|t}(0)=0$ and $x^*_{i|t}(0)=0$, which gives $u_\text{mpc}(0)=0$ as used in \eqref{equ:lemma2-part}. By \eqref{equ:mpc}, $v_\text{mpc}(x)$ is the sum of $\|x^*_{n_p|t}(x)\|_P^2$ and the terms $\|x^*_{i|t}(x)\|_Q^2$, $\|u^*_{i|t}(x)\|_R^2$, $i=0,\dots,n_p-1$. Each term has the form $f(x)^\top M f(x)$, where $M\in\{P,Q,R\}$ and $f(\cdot)$ is one of $x^*_{i|t}(\cdot)$, $u^*_{i|t}(\cdot)$. Hence $f(\cdot)$ is Lipschitz continuous on $\mathcal{F}$ with a constant $L_f$ from \eqref{equ:Lxu} and $f(0) = 0$. For $x_1,x_2\in\mathcal{F}$, we have
\begin{equation*}
\begin{aligned}
&|f(x_1)^\top M f(x_1) - f(x_2)^\top M f(x_2)| \\
 & \le|(f(x_1) - f(x_2))^\top M f(x_1)|\\
 & ~~ + |f(x_2)^\top M (f(x_1) - f(x_2))| \\
 & \le L_f \|M\| (\|x_1 - x_2\| \|f(x_1)\| + \|x_1 - x_2\| \|f(x_2)\|) \\
 & \le 2 L_f \|M\| \|x_1 - x_2\| \max\{\|f(x_1)\|, \|f(x_2)\|\} \\
 & \le 2 L_f^2 \|M\| \|x_1 - x_2\| \max\{\|x_1\|, \|x_2\|\},
\end{aligned}
\end{equation*}
where the last step uses $\|f(x_k)\|=\|f(x_k)-f(0)\|\le L_f\|x_k\|$, $k=1,2$. Summing these bounds over all terms, with $L_f=L_{x,n_p}$, $L_{x,i}$ and $L_{u,i}$ for $M=P$, $Q$ and $R$, respectively, yields
\begin{equation*}
\begin{aligned}
&|v_\text{mpc}(x_1) - v_\text{mpc}(x_2)| \\
&{}\le 2\left(L_{x,n_p}^2 \|P\| + \sum_{i=0}^{n_p-1} \big(L_{x,i} ^2\|Q\| + L_{u,i}^2 \|R\|\big) \right) \\
&~~~\times\|x_1 - x_2\| \max\{\|x_1\|, \|x_2\|\} \\
&:= c_0 \max\{\|x_1\|, \|x_2\|\} \|x_1 - x_2\|,
\end{aligned}
\end{equation*}
which is \eqref{newlc}. Finally, for any $x\neq0$ in the C-set $\mathcal{X}_\text{inv}$, \eqref{equ:exp_stability2}, \eqref{newlc} with $x_2=0$ and $v_\text{mpc}(0)=0$ give $c_1\|x\|^2\le v_\text{mpc}(x)\le c_0\|x\|^2$, so $c_0\ge c_1>0$.
\qed

\subsubsection{Proof of Lemma \ref{lem:convergence}}\label{sec:proofL3}
By \eqref{newlyapunova} and $\|x\|\le\sqrt{v_\text{mpc}(x)/c_1}$ from \eqref{equ:exp_stability2}, every $x\in\mathcal{X}_\text{inv}$ satisfies $v_\text{mpc}(Ax+Bu_\text{nn}(x))\le\psi(v_\text{mpc}(x))$, where
\begin{equation*}
\psi(w):=(1-c_3/c_2)w+c_4\sqrt{w/c_1}\,\delta_1+c_5\delta_1^2,\quad w\ge0.
\end{equation*}
We first show that $\psi(\cdot)$ is non-decreasing and concave. For any $x\neq0$ in the C-set $\mathcal{X}_\text{inv}$, \eqref{equ:exp_stability3} and $v_\text{mpc}(\cdot)\ge0$ give $c_3\|x\|^2\le v_\text{mpc}(x)$, and \eqref{equ:exp_stability2} gives $v_\text{mpc}(x)\le c_2\|x\|^2$. Hence $c_3\le c_2$, so the affine part of $\psi(\cdot)$ has slope $1-c_3/c_2\in[0,1)$. Since also $c_4\ge0$ and $\sqrt{w}$ is non-decreasing and concave, the claim follows.

Next, we find the fixed point of $\psi(\cdot)$. With $s=\sqrt{w}$, the difference $w-\psi(w)=(c_3/c_2)s^2-(c_4\delta_1/\sqrt{c_1})s-c_5\delta_1^2$ is a quadratic in $s$ with positive leading coefficient and negative constant term, as $\delta_1>0$ and $c_5=c_0\|B\|^2>0$ by Lemma \ref{lem:3Lip} and $B\neq0$. Hence, for $s\ge0$, it is negative before its unique positive root and positive after it. Therefore, $\psi(\cdot)$ has a unique fixed point $a_\infty>0$, and $\psi(w)<w$ if and only if $w>a_\infty$. Solving the quadratic gives
$$ a_\infty=\underbrace{\frac{c_2 \left( 2 c_1 c_3 c_5 + c_2 c_4^2 + c_4 \sqrt{ c_2 (4 c_1 c_3 c_5 + c_2 c_4^2) } \right)}{2 c_1 c_3^2}}_{c_6}\delta_1^2.$$
The constant $c_6>0$ depends only on $c_1,\ldots,c_5$. By assumption, $a_\infty=c_6\delta_1^2<\gamma$, where $\gamma$ is the level that defines $\mathcal{X}_\text{inv}$ in \eqref{sublevel}. Hence $\psi(\gamma)<\gamma$ and $\mathcal{V}(c_6\delta_1^2)\subseteq\mathcal{V}(\gamma)=\mathcal{X}_\text{inv}$. As $\psi(\cdot)$ is non-decreasing, every $x\in\mathcal{V}(c_6\delta_1^2)$ therefore satisfies $v_\text{mpc}(Ax+Bu_\text{nn}(x))\le\psi(v_\text{mpc}(x))\le\psi(a_\infty)=a_\infty$. Since $v_\text{mpc}(\cdot)$ is defined only on $\mathcal{F}\subseteq\mathcal{X}$, every state $z$ with $v_\text{mpc}(z)\le w$ belongs to $\mathcal{V}(w)$. Hence $\mathcal{V}(c_6\delta_1^2)$ is invariant.

Let $x_k$ be the closed-loop state at step $k$ with $x_0\in\mathcal{X}_\text{inv}$. To bound $v_\text{mpc}(x_k)$ from above, we define the recursive sequence
$$
\begin{dcases}
a_0 = \gamma, \\
a_{k+1} = \psi(a_k), \quad k\in\mathbb{N},
\end{dcases}
$$
and show by induction that $v_\text{mpc}(x_k)\le a_k$ and $a_\infty\le a_k\le\gamma$ for all $k$. This holds for $k=0$ because $x_0\in\mathcal{V}(\gamma)$ and $a_\infty<\gamma$. If it holds for $k$, then $x_k\in\mathcal{V}(\gamma)=\mathcal{X}_\text{inv}$. Hence, as $\psi(\cdot)$ is non-decreasing, $v_\text{mpc}(x_{k+1})\le\psi(v_\text{mpc}(x_k))\le\psi(a_k)=a_{k+1}$ and $a_\infty=\psi(a_\infty)\le\psi(a_k)=a_{k+1}\le\psi(\gamma)<\gamma$.

\begin{figure}
    \centering
    \includegraphics[width=1\linewidth]{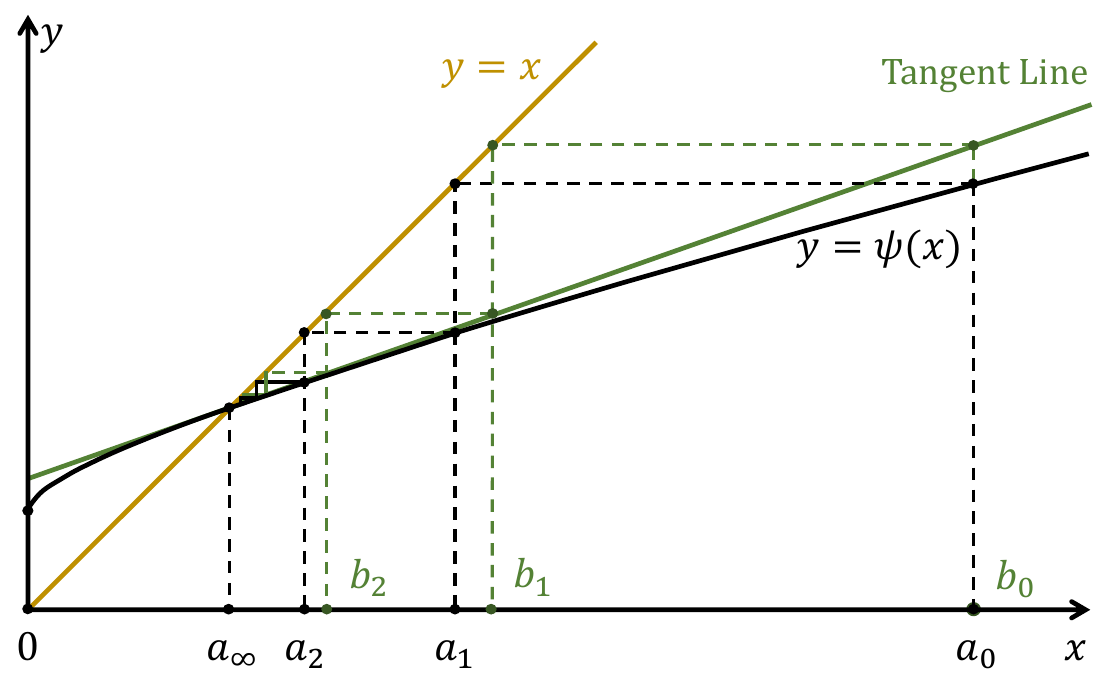}
    \caption{The recursive sequence $a_k$ and its convergence behavior. The black curve is $y=\psi(x)$, the yellow line is the diagonal $y=x$, and the green line is the tangent of $\psi(\cdot)$ at its fixed point $a_\infty$. }
    \label{fig:ak}
\end{figure}

Finally, we bound the rate at which $a_k$ approaches $a_\infty$. Since $a_\infty=c_6\delta_1^2$, the derivative is $\psi'(a_\infty)=1-c_3/c_2+c_4/(2\sqrt{c_1c_6})$. Dividing $\psi(a_\infty)=a_\infty$ by $a_\infty$ gives $c_3/c_2=c_4/\sqrt{c_1c_6}+c_5/c_6$, so $\psi'(a_\infty)$ equals the rate $q=1-c_3/(2c_2)-c_5/(2c_6)<1$ defined after Lemma \ref{lem:convergence}. Since $\psi(\cdot)$ is concave and differentiable at $a_\infty>0$, it lies below its tangent there, i.e., $\psi(w)-a_\infty\le q(w-a_\infty)$ for $w\ge a_\infty$. Taking $w=a_k\ge a_\infty$ gives $a_{k+1}-a_\infty\le q(a_k-a_\infty)$. Since $c_3\le c_2$ and $c_4\ge0$ give $q\ge0$, iterating yields $a_k-a_\infty\le q^k(\gamma-a_\infty)$. In Fig. \ref{fig:ak}, this bound is the sequence $b_k$ obtained by iterating the tangent from $b_0=a_0$. As $v_\text{mpc}(x_k)\le a_k$, we obtain $v_\text{mpc}(x_k)-c_6\delta_1^2\le a_k-a_\infty\le q^k\gamma$ for all $k$, i.e., the closed-loop system converges exponentially to $\mathcal{V}(c_6\delta_1^2)$.
\qed

\subsubsection{Projection onto tightened sets}\label{sec:proof_3Lip2}
We now present Lemma \ref{lem:3Lip2} regarding the properties of $r(\mathcal{U},\epsilon)$, defined in the proof of Lemma \ref{lem:projection}. In what follows, we denote the tightened set $\mathcal{U}'(\epsilon)$ defined in \eqref{tightset} as $\mathcal{T}(\mathcal{U}, \epsilon)$. Replacing $\mathcal{U}$ in \eqref{tightset} by any set $\mathcal{S}$ defines $\mathcal{T}(\mathcal{S},\epsilon)$, and we call $\mathcal{T}(\cdot,\cdot)$ the tightening operator.

\begin{lemma}\label{lem:3Lip2} For all $0\le\epsilon_1<\epsilon_2<d(\mathcal{U})$, $r(\mathcal{U},\epsilon_2)-r(\mathcal{U},\epsilon_1)\ge\epsilon_2-\epsilon_1$. In particular, $r(\mathcal{U},\cdot)$ is strictly increasing on $[0,d(\mathcal{U}))$. Moreover, $\epsilon \leq r(\mathcal{U}, \epsilon) \leq \epsilon D(\mathcal{U}) / d(\mathcal{U})$ for any $\epsilon\in[0,d(\mathcal{U}))$.
\end{lemma}

To prove Lemma \ref{lem:3Lip2}, we first introduce two fundamental properties of the tightening operator $\mathcal{T}(\cdot,\cdot)$.

\begin{lemma}\label{lem:3Lip1}
The tightened set $\mathcal{U}'(\epsilon)$ remains a C-set for every $\epsilon\in[0,d(\mathcal{U}))$.
\end{lemma}

\begin{proof}
Let $\epsilon\in[0,d(\mathcal{U}))$. By \eqref{tightset}, $\mathcal{T}(\mathcal{U},\epsilon) \subseteq \mathcal{U}$, which is a C-set, so $\mathcal{T}(\mathcal{U},\epsilon)$ is bounded. Since $0\in\mathcal{B}(\epsilon)$, \eqref{tightset} reads $\mathcal{T}(\mathcal{U},\epsilon)=\{u:\mathcal{B}(u,\epsilon)\subseteq\mathcal{U}\}$.

To prove convexity, let $u_1, u_2 \in \mathcal{T}(\mathcal{U},\epsilon)$, $\alpha \in [0,1]$ and $u_3 = \alpha u_1 + (1-\alpha) u_2$. For any $\Delta \in \mathcal{B}(\epsilon)$, $u_3 + \Delta = \alpha (u_1 + \Delta) + (1-\alpha) (u_2 + \Delta)$. Since $u_1 + \Delta, u_2 + \Delta \in \mathcal{U}$ and $\mathcal{U}$ is convex, $u_3 + \Delta \in \mathcal{U}$. As $\Delta$ is arbitrary, $u_3 \in \mathcal{T}(\mathcal{U},\epsilon)$. Therefore, $\mathcal{T}(\mathcal{U},\epsilon)$ is convex.

To prove closedness, let $u \notin \mathcal{T}(\mathcal{U},\epsilon)$. Then there exists $\Delta \in \mathcal{B}(\epsilon)$ such that $u + \Delta \notin \mathcal{U}$. Since $\mathcal{U}$ is closed, $\mathbb{R}^{n_u} \setminus \mathcal{U}$ is open. Hence, there exists $\epsilon'>0$ such that for all $\Delta' \in \mathcal{B}(\epsilon')$, $u + \Delta + \Delta' \notin \mathcal{U}$. Therefore, for all $u' \in \mathcal{B}(u,\epsilon')$, $u'+\Delta\notin\mathcal{U}$ and thus $u' \notin \mathcal{T}(\mathcal{U},\epsilon)$. Hence $\mathcal{T}(\mathcal{U},\epsilon)$ has an open complement, so it is closed. Being bounded, it is compact.

It remains to show that the origin is an interior point. Since $\mathcal{B}(\delta)\subseteq\mathcal{U}$ for every $\delta<d(\mathcal{U})$ and $\mathcal{U}$ is closed, $\mathcal{B}(d(\mathcal{U}))\subseteq\mathcal{U}$. For any $u\in\mathcal{B}(d(\mathcal{U})-\epsilon)$ and $\Delta\in\mathcal{B}(\epsilon)$, $u+\Delta\in\mathcal{B}(d(\mathcal{U}))\subseteq\mathcal{U}$. Hence $\mathcal{B}(d(\mathcal{U})-\epsilon)\subseteq\mathcal{T}(\mathcal{U},\epsilon)$. As $\epsilon<d(\mathcal{U})$, this ball has positive radius, so $\mathcal{T}(\mathcal{U},\epsilon)$ contains the origin in its interior. Being convex and compact, it is a C-set.
\end{proof}

\begin{lemma}\label{lem:3Lip3}
For any $\epsilon_1, \epsilon_2 \ge 0$, it holds that $\mathcal{T}(\mathcal{T}(\mathcal{U},\epsilon_1),\epsilon_2) = \mathcal{T}(\mathcal{U},\epsilon_1+\epsilon_2)$.
\end{lemma}

\begin{proof}
For any set $\mathcal{S}$ and $\epsilon\ge0$, since $0\in\mathcal{B}(\epsilon)$, a point $u$ belongs to $\mathcal{T}(\mathcal{S},\epsilon)$ if and only if $u+\Delta\in\mathcal{S}$ for all $\Delta\in\mathcal{B}(\epsilon)$. If $\epsilon_1=\epsilon_2=0$, both sides equal $\mathcal{U}$ because $\mathcal{B}(0)=\{0\}$. Let $\epsilon_1+\epsilon_2>0$.
We first prove $\mathcal{T}(\mathcal{T}(\mathcal{U},\epsilon_1),\epsilon_2) \subseteq \mathcal{T}(\mathcal{U},\epsilon_1 + \epsilon_2)$. Let $u \in \mathcal{T}(\mathcal{T}(\mathcal{U},\epsilon_1),\epsilon_2)$. By this characterization with $\mathcal{S}=\mathcal{T}(\mathcal{U},\epsilon_1)$, $u+\Delta_2\in\mathcal{T}(\mathcal{U},\epsilon_1)$ for all $\Delta_2\in\mathcal{B}(\epsilon_2)$. Applying it with $\mathcal{S}=\mathcal{U}$ to each $u+\Delta_2$ gives $u+\Delta_2+\Delta_1\in\mathcal{U}$ for all $\Delta_1\in\mathcal{B}(\epsilon_1)$ and $\Delta_2\in\mathcal{B}(\epsilon_2)$. Let $\Delta \in \mathcal{B}(\epsilon_1 + \epsilon_2)$. The vectors $\Delta_1 = \epsilon_1 \Delta / (\epsilon_1 + \epsilon_2)$ and $\Delta_2 = \epsilon_2 \Delta / (\epsilon_1 + \epsilon_2)$ lie in $\mathcal{B}(\epsilon_1)$ and $\mathcal{B}(\epsilon_2)$, respectively, and sum to $\Delta$. Hence $u+\Delta=u+\Delta_2+\Delta_1\in\mathcal{U}$. As $\Delta$ is arbitrary, $u \in \mathcal{T}(\mathcal{U},\epsilon_1 + \epsilon_2)$.

Conversely, we prove $\mathcal{T}(\mathcal{U},\epsilon_1 + \epsilon_2) \subseteq \mathcal{T}(\mathcal{T}(\mathcal{U},\epsilon_1),\epsilon_2)$. Let $u \in \mathcal{T}(\mathcal{U},\epsilon_1 + \epsilon_2)$, $\Delta_1 \in \mathcal{B}(\epsilon_1)$ and $\Delta_2 \in \mathcal{B}(\epsilon_2)$. Since $\Delta_1 + \Delta_2 \in \mathcal{B}(\epsilon_1 + \epsilon_2)$, it follows that $u+\Delta_2+\Delta_1\in\mathcal{U}$. As $\Delta_1\in\mathcal{B}(\epsilon_1)$ is arbitrary, $u+\Delta_2\in\mathcal{T}(\mathcal{U},\epsilon_1)$. As $\Delta_2\in\mathcal{B}(\epsilon_2)$ is arbitrary, $u\in\mathcal{T}(\mathcal{T}(\mathcal{U},\epsilon_1),\epsilon_2)$. Therefore, $\mathcal{T}(\mathcal{T}(\mathcal{U},\epsilon_1),\epsilon_2) = \mathcal{T}(\mathcal{U},\epsilon_1 + \epsilon_2)$.
\end{proof}

\begin{proof}[Proof of Lemma \ref{lem:3Lip2}]
By Lemma \ref{lem:3Lip1}, $\mathcal{T}(\mathcal{U},\epsilon)$ is a C-set, and in particular nonempty, closed and convex, for every $\epsilon\in[0,d(\mathcal{U}))$. Hence, for every $u\in\mathbb{R}^{n_u}$, the projection $\Pi_{\mathcal{T}(\mathcal{U},\epsilon)}(u)$ is well defined and is the point of $\mathcal{T}(\mathcal{U},\epsilon)$ closest to $u$. Since the distance $\|u-\Pi_{\mathcal{T}(\mathcal{U},\epsilon)}(u)\|$ from $u$ to $\mathcal{T}(\mathcal{U},\epsilon)$ is continuous in $u$ and $\mathcal{U}$ is compact, the maximum defining $r(\mathcal{U},\epsilon)$ is attained. Moreover, $r(\mathcal{U},0)=0$ because $\mathcal{T}(\mathcal{U},0)=\mathcal{U}$.

We first prove $r(\mathcal{U},\epsilon)\ge\epsilon$ for every $\epsilon\in(0,d(\mathcal{U}))$ by showing that a boundary point of $\mathcal{U}$ is at distance at least $\epsilon$ from $\mathcal{T}(\mathcal{U},\epsilon)$. As $\mathcal{U}$ is nonempty and compact, there exists $\bar u\in\partial\mathcal{U}\subseteq\mathcal{U}$. As $\mathcal{U}$ is convex, the supporting hyperplane theorem yields a unit vector $n$ such that $\langle n,u\rangle\le\langle n,\bar u\rangle$ for all $u\in\mathcal{U}$, where $\langle x,y\rangle:=x^\top y$. Let $z\in\mathcal{T}(\mathcal{U},\epsilon)$. Since $\epsilon n\in\mathcal{B}(\epsilon)$, \eqref{tightset} gives $z+\epsilon n\in\mathcal{U}$. The choice of $n$ then gives $\langle n,z\rangle+\epsilon=\langle n,z+\epsilon n\rangle\le\langle n,\bar u\rangle$. Combined with the Cauchy--Schwarz inequality, this yields $\|\bar u-z\|\ge\langle n,\bar u-z\rangle\ge\epsilon$. With $z=\Pi_{\mathcal{T}(\mathcal{U},\epsilon)}(\bar u)$, the definition of $r(\mathcal{U},\epsilon)$ gives $r(\mathcal{U},\epsilon)\ge\|\bar u-\Pi_{\mathcal{T}(\mathcal{U},\epsilon)}(\bar u)\|\ge\epsilon$.

Next, let $0\le\epsilon_1<\epsilon_2<d(\mathcal{U})$. We prove $r(\mathcal{U},\epsilon_2)-r(\mathcal{U},\epsilon_1)\ge\epsilon_2-\epsilon_1$. If $\epsilon_1=0$, this inequality follows from $r(\mathcal{U},0)=0$ and $r(\mathcal{U},\epsilon_2)\ge\epsilon_2$. Now let $\epsilon_1>0$, and let
\begin{equation*}
u_1 \in \arg\max_{u \in \mathcal{U}} \|u - \Pi_{\mathcal{T}(\mathcal{U}, \epsilon_1)}(u)\|.
\end{equation*}
The lower bound proved above gives $\|u_1-\Pi_{\mathcal{T}(\mathcal{U},\epsilon_1)}(u_1)\|=r(\mathcal{U},\epsilon_1)\ge\epsilon_1>0$. Hence, $u_1\notin\mathcal{T}(\mathcal{U},\epsilon_1)$. By Lemma \ref{lem:3Lip3}, $\mathcal{T}(\mathcal{U}, \epsilon_2) = \mathcal{T}(\mathcal{T}(\mathcal{U}, \epsilon_1), \epsilon_2 - \epsilon_1)$. By \eqref{tightset} applied to the set $\mathcal{T}(\mathcal{U},\epsilon_1)$, this means that every point of $\mathcal{T}(\mathcal{U},\epsilon_2)$ is the center of a ball of radius $\epsilon_2-\epsilon_1$ contained in $\mathcal{T}(\mathcal{U},\epsilon_1)$. For the point $\Pi_{\mathcal{T}(\mathcal{U},\epsilon_2)}(u_1)\in\mathcal{T}(\mathcal{U},\epsilon_2)$, this reads
\begin{equation}\label{equ:u1}
\mathcal{B}(\Pi_{\mathcal{T}(\mathcal{U}, \epsilon_2)}(u_1),\epsilon_2 - \epsilon_1)\subseteq\mathcal{T}(\mathcal{U}, \epsilon_1).
\end{equation}
Since $u_1\notin\mathcal{T}(\mathcal{U},\epsilon_1)$, \eqref{equ:u1} implies that $u_1$ lies outside the ball on its left-hand side, i.e., $\|u_1-\Pi_{\mathcal{T}(\mathcal{U},\epsilon_2)}(u_1)\|>\epsilon_2-\epsilon_1>0$. Hence, the vector
\begin{equation*}
v := (\epsilon_2 - \epsilon_1)\dfrac{u_1 - \Pi_{\mathcal{T}(\mathcal{U}, \epsilon_2)}(u_1)}{\|u_1 - \Pi_{\mathcal{T}(\mathcal{U}, \epsilon_2)}(u_1)\|}
\end{equation*}
is well defined. Since $\|v\|=\epsilon_2-\epsilon_1$, \eqref{equ:u1} yields
\begin{equation}\label{equ:u1v}
\Pi_{\mathcal{T}(\mathcal{U}, \epsilon_2)}(u_1) + v \in \mathcal{T}(\mathcal{U}, \epsilon_1).
\end{equation}
Because $\Pi_{\mathcal{T}(\mathcal{U},\epsilon_1)}(u_1)$ is the point of $\mathcal{T}(\mathcal{U},\epsilon_1)$ closest to $u_1$, \eqref{equ:u1v} gives the first inequality in
\begin{equation}\label{equ:r}
\begin{aligned}
r(\mathcal{U}, \epsilon_1) &= \|u_1 - \Pi_{\mathcal{T}(\mathcal{U}, \epsilon_1)}(u_1)\|\\
&{}\le \|u_1 - (\Pi_{\mathcal{T}(\mathcal{U}, \epsilon_2)}(u_1) + v)\| \\
&= -\epsilon_2 + \epsilon_1 + \|u_1 - \Pi_{\mathcal{T}(\mathcal{U}, \epsilon_2)}(u_1)\| \\
&\le -\epsilon_2 + \epsilon_1 + r(\mathcal{U}, \epsilon_2).
\end{aligned}
\end{equation}
In \eqref{equ:r}, the first equality follows from the choice of $u_1$. The second equality holds because $v$ has the same direction as $u_1-\Pi_{\mathcal{T}(\mathcal{U},\epsilon_2)}(u_1)$ and a smaller norm. The last inequality follows from the definition of $r(\mathcal{U},\epsilon_2)$. Hence, $r(\mathcal{U},\epsilon_2)-r(\mathcal{U},\epsilon_1)\ge\epsilon_2-\epsilon_1$ also holds for $\epsilon_1>0$. Since $\epsilon_2-\epsilon_1>0$, it follows that $r(\mathcal{U},\cdot)$ is strictly increasing on $[0,d(\mathcal{U}))$.

Finally, we prove $r(\mathcal{U},\epsilon)\le\epsilon D(\mathcal{U})/d(\mathcal{U})$ for every $\epsilon\in(0,d(\mathcal{U}))$ by showing that $(1-\mu)u\in\mathcal{T}(\mathcal{U},\epsilon)$ for every $u\in\mathcal{U}$, where $\mu:=\epsilon/d(\mathcal{U})\in(0,1)$. Let $u\in\mathcal{U}$ and $\Delta\in\mathcal{B}(\epsilon)$. Then $\Delta/\mu\in\mathcal{B}(d(\mathcal{U}))\subseteq\mathcal{U}$ (cf. the proof of Lemma \ref{lem:3Lip1}). Since $\mathcal{U}$ is convex, $(1-\mu)u+\Delta=(1-\mu)u+\mu(\Delta/\mu)\in\mathcal{U}$. As this holds for every $\Delta\in\mathcal{B}(\epsilon)$, including $\Delta=0$, \eqref{tightset} yields $(1-\mu)u\in\mathcal{T}(\mathcal{U},\epsilon)$. Because $\Pi_{\mathcal{T}(\mathcal{U},\epsilon)}(u)$ is the point of $\mathcal{T}(\mathcal{U},\epsilon)$ closest to $u$, we obtain $\|u-\Pi_{\mathcal{T}(\mathcal{U},\epsilon)}(u)\|\le\|u-(1-\mu)u\|=\mu\|u\|\le\epsilon D(\mathcal{U})/d(\mathcal{U})$, where the last step uses $\|u\|\le D(\mathcal{U})$. As $u\in\mathcal{U}$ is arbitrary, $r(\mathcal{U},\epsilon)\le\epsilon D(\mathcal{U})/d(\mathcal{U})$. Combined with $r(\mathcal{U},\epsilon)\ge\epsilon$ and $r(\mathcal{U},0)=0$, this proves $\epsilon\le r(\mathcal{U},\epsilon)\le\epsilon D(\mathcal{U})/d(\mathcal{U})$ for every $\epsilon\in[0,d(\mathcal{U}))$.
\end{proof}

\subsubsection{Proof of Theorem \ref{thm:main}}\label{sec:proof_main}
The proof verifies the conditions of Lemmas \ref{lem:projection} and \ref{lem:convergence}, proves that the next state is feasible, and constructs a network of the stated size by Lemma \ref{lem:3nn}.

Lemma \ref{lem:convergence} requires $c_6\delta_1^2<\gamma$, i.e., $\delta_1<c_7\sqrt{\gamma}$ with $c_7:=1/\sqrt{c_6}$. We take the tightening parameter of Lemma \ref{lem:projection} as
\begin{equation}\label{equ:eps}
\epsilon:=\delta_1 d(\mathcal{U})/(2D(\mathcal{U})),
\end{equation}
which Lemma \ref{lem:projection} requires to satisfy $\epsilon<d(\mathcal{U})$, i.e., $\delta_1<2D(\mathcal{U})$. Both conditions hold for all $\delta_1\in(0,\bar\delta]$ with the conservative choice
\begin{equation}\label{equ:delta1}
\bar\delta = \min\left\{c_7d(\mathcal{U})\sqrt{\gamma}/(2 D(\mathcal{U})), d(\mathcal{U})\right\},
\end{equation}
since $d(\mathcal{U})\le D(\mathcal{U})$ bounds the two entries of the minimum by $c_7\sqrt{\gamma}/2<c_7\sqrt{\gamma}$ and $D(\mathcal{U})<2D(\mathcal{U})$, respectively.

Let $\delta_1\in(0,\bar\delta]$ and let $u_\text{nn}$ satisfy \eqref{tightedapprox} with \eqref{equ:eps}. By Lemma \ref{lem:projection}, \eqref{equ:con2} and \eqref{equ:con1} hold. To prove \eqref{newlyapunova} and \eqref{equ:con3}, fix $x\in\mathcal{X}_\text{inv}$ and let $z_\sigma:=Ax+Bu_\sigma$ with $u_\sigma:=u_\text{mpc}(x)+\sigma(u_\text{nn}(x)-u_\text{mpc}(x))$, $\sigma\in[0,1]$. By \eqref{equ:exp_stability1}, $z_0=Ax+Bu_\text{mpc}(x)\in\mathcal{X}_\text{inv}$ is feasible. Since $\|u_\sigma-u_\text{mpc}(x)\|\le\delta_1$ by \eqref{equ:con1}, \eqref{newlyapunova} holds with $u_\sigma$ in place of $u_\text{nn}(x)$ whenever $z_\sigma$ is also feasible, as Lemma \ref{lem:3Lip} requires. As in the proof of Lemma \ref{lem:convergence}, every feasible $z_\sigma$ thus satisfies $v_\text{mpc}(z_\sigma)\le\psi(v_\text{mpc}(x))$, where, by that proof, $\psi(\cdot)$ is non-decreasing and $\psi(\gamma)<\gamma$ as $c_6\delta_1^2<\gamma$. Since $v_\text{mpc}(x)\le\gamma$, each such $z_\sigma$ satisfies $v_\text{mpc}(z_\sigma)\le\psi(\gamma)<\gamma$, i.e., $z_\sigma\in\mathcal{V}(\gamma)$. Let $\sigma^*$ be the largest $\sigma\in[0,1]$ with $z_\sigma$ in the feasible set $\mathcal{F}$ of \eqref{equ:mpc}, which exists because $\mathcal{F}$ is closed (cf. the proof of Lemma \ref{lem:3Lip}) and $z_0\in\mathcal{F}$. Since $\gamma<\gamma_{\max}$ (cf. \eqref{sublevel}), $z_{\sigma^*}\in\mathcal{V}(\gamma)$ is an interior point of $\mathcal{F}$, so the maximality of $\sigma^*$ forces $\sigma^*=1$. Hence $Ax+Bu_\text{nn}(x)=z_1$ is feasible and lies in $\mathcal{V}(\gamma)=\mathcal{X}_\text{inv}$, which proves \eqref{newlyapunova} at $x$ and \eqref{equ:con3}. As $x$ is arbitrary and $c_6\delta_1^2<\gamma$, Lemma \ref{lem:convergence} gives, for any $x_0\in\mathcal{X}_\text{inv}$, exponential convergence to the invariant set $\mathcal{V}(c_6\delta_1^2)$.

Since the projection onto the C-set $\mathcal{T}(\mathcal{U},\epsilon)$ (cf. Lemma \ref{lem:3Lip1}) is non-expansive and $L_u=\mathcal{L}(u_\text{mpc},\mathcal{X}_\text{inv})$,
\begin{equation}\label{equ:proj_Lip}
\mathcal{L}(\Pi_{\mathcal{T}(\mathcal{U},\epsilon)}\circ u_\text{mpc},\mathcal{X}_\text{inv}) \le L_u.
\end{equation}
Define $f:=\Pi_{\mathcal{T}(\mathcal{U},\epsilon)}\circ u_\text{mpc}$. Since $|f_j(x)-f_j(y)|\le\|f(x)-f(y)\|$, \eqref{equ:proj_Lip} gives, for each component $f_j$, all $t\ge0$ and the modulus of continuity $\omega$ defined in Lemma \ref{lem:3nn},
\begin{equation}\label{equ:Lip_f}
\omega^{\mathcal{X}_\text{inv}}_{f_j}\big(2D(\mathcal{X}_\text{inv})t\big) \le 2D(\mathcal{X}_\text{inv})L_u\,t.
\end{equation}

Each $f_j$ is continuous by \eqref{equ:Lip_f} and thus bounded on the compact set $\mathcal{X}_\text{inv}$. By Lemma \ref{lem:3nn} with $\mathcal{E}=\mathcal{X}_\text{inv}$, $N'=n_w'$, $L'=n_d'$ and $p=\infty$, for which $(2D(\mathcal{E}))^{n_x/p}=1$ in \eqref{equ:3nn}, we construct $n_u$ ReLU NNs $\phi_j$, each approximating one component $f_j$. Since $f_j-\phi_j$ is continuous and $\mathcal{X}_\text{inv}$, as a C-set, is the closure of its interior, the $L^\infty$ bound \eqref{equ:3nn} holds pointwise on $\mathcal{X}_\text{inv}$. If each $\phi_j$ achieves the error bound $\epsilon/\sqrt{n_u}$, the stacked network $u_\text{nn}:=(\phi_1,\dots,\phi_{n_u})$ satisfies \eqref{tightedapprox}, since then $\|u_\text{nn}(x)-f(x)\|\le\sqrt{n_u}\,\epsilon/\sqrt{n_u}=\epsilon$. Imposing this error bound on the RHS of \eqref{equ:3nn}, we find that it holds if
\begin{equation*}
\begin{aligned}
&131\sqrt{n_x}~\omega^{\mathcal{X}_\text{inv}}_{f_j}(2D(\mathcal{X}_\text{inv})(n_w'^2n_d'^2\log_3(n_w'+2))^{-1/{n_x}})\\
&{}\le262\sqrt{n_x}D(\mathcal{X}_\text{inv})L_u(n_w'^2n_d'^2\log_3(n_w'+2))^{-1/n_x}\\
&\le\delta_1 d(\mathcal{U})/(2D(\mathcal{U})\sqrt{n_u}),
\end{aligned}
\end{equation*}
where the first inequality follows from \eqref{equ:Lip_f} and the second is equivalent to \eqref{equ:main}. Hence, under \eqref{equ:main}, the stacked network satisfies \eqref{tightedapprox}, so \eqref{equ:con} and the convergence shown above hold for it.

By Lemma \ref{lem:3nn}, each $\phi_j$ has width $3^{n_x+3}\max\{n_x\lfloor n_w'^{1/n_x}\rfloor, n_w'+2\}$ and $11n_d'+18+2n_x$ hidden layers, i.e., $11n_d'+19+2n_x$ affine layers as counted by $n_d$ in \eqref{equ:nn}. Stacking the $n_u$ networks in parallel multiplies the width by $n_u$ and leaves the depth unchanged, which gives \eqref{equ:width} and \eqref{equ:depth}.
\qed

\subsection{Proof of Results in Section \ref{sec:pro2}}\label{sec:appB}
Theorem \ref{thm:main2} is proved as Theorem \ref{thm:main}, except that the network approximates $\widetilde{u}_\text{mpc}(\cdot)$ in \eqref{equ:transformation} on $T(\mathcal{X}_\text{inv})$ and that the convergence argument treats states far from and near the origin separately. This requires $T(\cdot)$ to be invertible, $T(\mathcal{X}_\text{inv})$ to be bounded and the GLC of $\widetilde{u}_\text{mpc}(\cdot)$ to be at most $L_u$, which Appendices \ref{sec:proof_T} and \ref{sec:proof_3Lip6} establish.

\subsubsection{Properties of the input scaling $T(\cdot)$}\label{sec:proof_T}
We first derive the explicit expression \eqref{equ:Texplicit} of $T(\cdot)$, on which Lemmas \ref{lem:T_bijective}--\ref{lem:3Lip4} below rely. For $x\neq0$ and $s\in[0,1]$, \eqref{equ:relative_error} gives $\beta(sx)=\underline\delta/\bar\delta$ if $s\|x\|\le\underline\delta/c_8$, $\beta(sx)=c_8s\|x\|/\bar\delta$ if $\underline\delta/c_8<s\|x\|<\bar\delta/c_8$, and $\beta(sx)=1$ otherwise. We insert this into \eqref{equ:Tdef} and write $\tau_1:=\underline\delta/(c_8\|x\|)$ and $\tau_2:=\bar\delta/(c_8\|x\|)$ for the values of $s$ at which $s\|x\|$ equals $\underline\delta/c_8$ and $\bar\delta/c_8$, respectively.

We consider three cases, defined by regions $\Omega_1,\Omega_2,\Omega_3$ that cover the ball $\mathcal{X}_D:=\mathcal{B}(R_D)$ with $R_D:=\max\{D(\mathcal{X}_\text{inv}),\bar\delta/c_8\}$.

For any $x\in\Omega_1:=\mathcal{B}(\underline\delta/c_8)$, the integrand in \eqref{equ:Tdef} equals $2\bar\delta/\underline\delta$ as $s\|x\|\le\underline\delta/c_8\le\bar\delta/c_8$, so $T_1(x)=2\bar\delta x/\underline\delta$, where $T_1$ is the restriction of $T(\cdot)$ to $\Omega_1$. Hence $T_1$ maps $\Omega_1$ onto $\Omega_1':=\mathcal{B}(2\bar\delta/c_8)$, with inverse $T_1^{-1}(y)=\underline\delta y/(2\bar\delta)$ for $y\in\Omega_1'$.

For any $x\in\Omega_2:=\mathrm{\mathbf{cl}}(\mathcal{B}(\bar\delta/c_8)\setminus\mathcal{B}(\underline\delta/c_8))$, the indicator in \eqref{equ:Tdef} equals one, and we have
$$
\begin{aligned}
T_2(x)&=x\int_0^{1}\dfrac{2}{\beta(sx)}\mathrm{d}s\\
&=x\int_0^{\tau_1}\dfrac{2\bar\delta}{\underline\delta}\mathrm{d}s+x\int_{\tau_1}^{1}\dfrac{2\bar\delta}{c_8s\|x\|}\mathrm{d}s\\
&=\underbrace{\dfrac{2\bar\delta}{c_8}\left(1+\ln\left(\dfrac{\|x\|}{\underline\delta/c_8}\right)\right)}_{\rho_2(\|x\|)} \dfrac{x}{\|x\|}.
\end{aligned}
$$

For any $x\in\Omega_3:=\mathrm{\mathbf{cl}}(\mathcal{X}_D\setminus\mathcal{B}(\bar\delta/c_8))$, and more generally for any $x$ with $\|x\|\ge\bar\delta/c_8$, we have
$$
\begin{aligned}
T_3(x)&=x\int_0^{\tau_2}\dfrac{2}{\beta(sx)}\mathrm{d}s+x\int_{\tau_2}^{1}\dfrac{1}{\beta(sx)}\mathrm{d}s\\
&=\dfrac{2\bar\delta}{c_8}\left(1+\ln\left(\dfrac{\bar\delta}{\underline\delta}\right)\right)\dfrac{x}{\|x\|}+(1-\tau_2)x\\
&=x + \underbrace{\dfrac{\bar\delta}{c_8} \left(1+2\ln \left(\dfrac{\bar\delta}{\underline\delta}\right)\right)}_{\bar\rho_3} \dfrac{x}{\|x\|}.
\end{aligned}
$$
The first line holds because the indicator in \eqref{equ:Tdef} equals one exactly for $s\le\tau_2$. In the second line, the first integral is split at $\tau_1$ as for $T_2(x)$, with upper limit $\tau_2$ instead of $1$, and $\beta(sx)=1$ in the second integral.

By \eqref{equ:Tdef}, $T(0)=0$. For $x\neq0$, the three formulas give $T(x)=\rho(\|x\|)x/\|x\|$, where $\rho:[0,\infty)\to[0,\infty)$ equals $2\bar\delta s/\underline\delta$, $\rho_2(s)$ and $s+\bar\rho_3$ for $s$ in $[0,\underline\delta/c_8]$, $[\underline\delta/c_8,\bar\delta/c_8]$ and $[\bar\delta/c_8,\infty)$, respectively. Adjacent expressions agree at their common end point: the first two equal $2\bar\delta/c_8$ at $s=\underline\delta/c_8$, and the last two equal $(2\bar\delta/c_8)(1+\ln(\bar\delta/\underline\delta))$ at $s=\bar\delta/c_8$. Hence $\rho$ is well defined and continuous. The derivatives of these three expressions are $2\bar\delta/\underline\delta$, $2\bar\delta/(c_8s)$ and $1$. They are at least $1$ because $\underline\delta\le\bar\delta$ and $s\le\bar\delta/c_8$ on the second interval. Since $\rho$ is continuous, integrating $\rho'$ from $s_1$ to $s_2$ gives
\begin{equation}\label{equ:rhoexp}
\rho(s_2)-\rho(s_1)\ge s_2-s_1\ge0,\quad 0\le s_1\le s_2.
\end{equation}
In particular, $\rho$ is strictly increasing, and, as $\rho(0)=0$, taking $s_1=0$ gives $\rho(s)\ge s$.

Lemmas \ref{lem:T_bijective} and \ref{lem:3Lip5} below use \eqref{equ:rhoexp} to show that $T(\cdot)$ is invertible on $\mathcal{X}_D$ and that $T(\mathcal{X}_\text{inv})$ is bounded.

\begin{lemma}\label{lem:T_bijective}
The map $T(\cdot)$ is a continuous bijection of $\mathcal{X}_D$ onto the ball $T(\mathcal{X}_D)=\mathcal{B}(\rho(R_D))$, and $T^{-1}(\cdot)$ is $1$-Lipschitz continuous.
\end{lemma}
\begin{proof}
We first show that $\|T(x_1)-T(x_2)\|\ge\|x_1-x_2\|$ for all $x_1,x_2\in\mathcal{X}_D$. If one of them, say $x_2$, is $0$, then $\|T(x_1)-T(x_2)\|=\rho(\|x_1\|)\ge\|x_1\|=\|x_1-x_2\|$. Otherwise, let $s_i=\|x_i\|>0$, $r_i=\rho(s_i)$ and $c=\langle x_1,x_2\rangle/(s_1s_2)\in[-1,1]$. Then $T(x_i)=r_ix_i/s_i$, so
$$
\begin{aligned}
\|T(x_1)-T(x_2)\|^2-\|x_1-x_2\|^2={}&r_1^2-s_1^2+r_2^2-s_2^2\\
&-2(r_1r_2-s_1s_2)c.
\end{aligned}
$$
Since $r_i\ge s_i$, the RHS is non-increasing in $c$. As $c\le1$, the RHS is at least its value at $c=1$, namely $(r_1-r_2)^2-(s_1-s_2)^2$, which is non-negative by \eqref{equ:rhoexp}.

Consequently, any two distinct points of $\mathcal{X}_D$, whether in the same region $\Omega_i$ or in different ones, have distinct images, i.e., $T(\cdot)$ is injective on $\mathcal{X}_D$. Applied to $x_i=T^{-1}(y_i)$, the same inequality shows that $T^{-1}(\cdot)$ is $1$-Lipschitz continuous on $T(\mathcal{X}_D)$. The map $T(\cdot)$ is continuous on $\mathcal{X}_D\setminus\{0\}$ because $\rho$ is continuous, and at the origin because $\|T(x)\|=\rho(\|x\|)\to\rho(0)=0$ as $x\to0$. Since $\rho$ is increasing, $T(\mathcal{X}_D)\subseteq\mathcal{B}(\rho(R_D))$. Conversely, for any $y\in\mathcal{B}(\rho(R_D))\setminus\{0\}$, the intermediate value theorem gives $s\in(0,R_D]$ with $\rho(s)=\|y\|$, and then $T(sy/\|y\|)=y$. Together with $T(0)=0$, this gives $T(\mathcal{X}_D)=\mathcal{B}(\rho(R_D))$.
\end{proof}

\begin{lemma}\label{lem:3Lip5}
$T(\mathcal{X}_\text{inv})\subseteq\mathcal{B}(D(\mathcal{X}_\text{inv})+(\bar\delta/c_8)(1+2 \ln(\bar\delta/\underline\delta)))$.
\end{lemma}
\begin{proof}
As shown before \eqref{equ:rhoexp}, $T(x)=\rho(\|x\|)x/\|x\|$ for $x\neq0$ and $T(0)=0$. Since $\rho(0)=0$, it follows that $\|T(x)\|=\rho(\|x\|)$ for every $x\in\mathbb{R}^{n_x}$. By \eqref{equ:rhoexp}, $\rho(s)-s$ is non-decreasing in $s$. For $s\ge\bar\delta/c_8$, it equals $\bar\rho_3=(\bar\delta/c_8)(1+2\ln(\bar\delta/\underline\delta))$, the offset in the case $x\in\Omega_3$, where $T_3(x)=x+\bar\rho_3x/\|x\|$. Hence $\|T(x)\|=\rho(\|x\|)\le\|x\|+\bar\rho_3$ for every $x\in\mathbb{R}^{n_x}$.
Since $\|x\|\le D(\mathcal{X}_\text{inv})$ for any $x\in\mathcal{X}_\text{inv}$, we obtain $T(\mathcal{X}_\text{inv})\subseteq\mathcal{B}(D(\mathcal{X}_\text{inv})+\bar\rho_3)$.
\end{proof}

Lemma \ref{lem:3Lip6} needs more than Lemma \ref{lem:T_bijective}: near the origin, dividing $u_\text{mpc}(\cdot)$ by $\beta(\cdot)$ multiplies its Lipschitz constant by $1/\beta(\cdot)$, so $T(\cdot)$ must stretch distances by at least this factor. Lemma \ref{lem:3Lip4} shows this locally, up to a second-order term.

\begin{lemma}[State-dependent Lipschitz continuity of $T^{-1}(\cdot)$]\label{lem:3Lip4}
Let $\eta_\text{max}= \underline\delta/(2c_8)$ and $c_{12}=16c_8\bar\delta/\underline\delta^2$. For any $\eta\in(0,\eta_\text{max})$, $i=1,2,3$, $x_0\in\Omega_i$, $x_1,x_2\in \mathcal{B}(x_0,\eta)\cap\Omega_i$, we have
\begin{equation}\label{equ:3Lip4}
\|x_1-x_2\|/\beta_\eta(x_0)\le b\|T(x_1)-T(x_2)\|+c_{12}\eta^2,
\end{equation}
where $\Omega_1,\Omega_2,\Omega_3$ are defined above, $b = 1/(1 + \mathbbm{1}(i\le2))$, and $\beta_\eta(x_0):=\min_{x\in\mathcal{B}(x_0,\eta)}\beta(x)$.
\end{lemma}
\begin{proof}
We first bound the variation of $1/\beta(\cdot)$ and then divide the proof into three cases by the regions $\Omega_1,\Omega_2,\Omega_3$. By \eqref{equ:relative_error}, $1/\beta(x)=\bar\delta/\delta_2(x)=\Pi_{[1,\bar\delta/\underline\delta]}\big(\bar\delta/(c_8\|x\|)\big)$ for $x\neq0$, and $1/\beta(0)=\bar\delta/\underline\delta$. Hence $1/\beta(x)$ is non-increasing in $\|x\|$. For $r\ge\underline\delta/(2c_8)$, the derivative of $\bar\delta/(c_8r)$ has magnitude $\bar\delta/(c_8r^2)\le c_{11}:=4c_8\bar\delta/\underline\delta^2$. Since the projection is non-expansive, this bound gives $|1/\beta(x)-1/\beta(x')|\le c_{11}\big|\|x\|-\|x'\|\big|$ whenever $\|x\|,\|x'\|\ge\underline\delta/(2c_8)$. Now let $x_0\in\Omega_2\cup\Omega_3$, so that $\|x_0\|\ge\underline\delta/c_8$. Every $x\in\mathcal{B}(x_0,\eta)$ satisfies $\big|\|x\|-\|x_0\|\big|\le\eta$ and, since $\eta<\eta_\text{max}$, $\|x\|\ge\|x_0\|-\eta>\underline\delta/(2c_8)$. Hence $1/\beta_\eta(x_0)=\max_{x\in\mathcal{B}(x_0,\eta)}1/\beta(x)\le1/\beta(x_0)+c_{11}\eta$. As $\|x_1-x_2\|\le2\eta$ for $x_1,x_2\in\mathcal{B}(x_0,\eta)$, we obtain
$$
\|x_1-x_2\|/\beta_\eta(x_0)\le\|x_1-x_2\|/\beta(x_0)+2c_{11}\eta^2.
$$

For any $x_0,x_1,x_2\in \Omega_1$, since $T_1(x)=2\bar\delta x/\underline\delta$ and $\beta(x_0)=\underline\delta/\bar\delta$, we have $\|T_1(x_1)-T_1(x_2)\|=2\bar\delta\|x_1-x_2\|/\underline\delta=2\|x_1-x_2\|/\beta(x_0)$.
Moreover, $\beta_\eta(x_0)=\beta(x_0)$ because $\underline\delta/\bar\delta$ is the minimum value of $\beta$. Hence \eqref{equ:3Lip4} holds with $b=1/2$ for any $x_0,x_1,x_2\in \Omega_1$.

For any $x_0,x_1,x_2\in \Omega_3$, since $T=T_3$ on $\Omega_3$, we have
\begin{equation}\label{equ:3Lip5}
\begin{aligned}
&\|T(x_2)-T(x_1)\|^2\\
&=\left\|x_2+\bar\rho_3\dfrac{x_2}{\|x_2\|}-x_1-\bar\rho_3\dfrac{x_1}{\|x_1\|}\right\|^2\\
&=\,\|x_2-x_1\|^2 + 2\bar\rho_3\left\langle x_2-x_1,\dfrac{x_2}{\|x_2\|}-\dfrac{x_1}{\|x_1\|}\right\rangle \\
&~~~~~~ + \bar\rho_3^2\left\|\dfrac{x_2}{\|x_2\|}-\dfrac{x_1}{\|x_1\|}\right\|^2\\
&\ge \|x_2-x_1\|^2 + 2\bar\rho_3\left\langle x_2-x_1,\dfrac{x_2}{\|x_2\|}-\dfrac{x_1}{\|x_1\|}\right\rangle\\
&=\|x_2-x_1\|^2\\
&~~~~~~ + 2\bar\rho_3\left(\|x_2\|+\|x_1\|\right)\left(1-\left\langle \dfrac{x_1}{\|x_1\|},\dfrac{x_2}{\|x_2\|}\right\rangle\right)\\
&\ge \, \|x_2-x_1\|^2=\|x_2-x_1\|^2/\beta^2(x_0)
\end{aligned}
\end{equation}
where the last line uses $\bar\rho_3>0$, $\langle x_1/\|x_1\|,x_2/\|x_2\|\rangle\le1$ and $\beta(x_0)=1$. Thus, by \eqref{equ:3Lip5}, the first display of this proof and $2c_{11}<c_{12}$, \eqref{equ:3Lip4} holds with $b=1$ for any $x_0\in \Omega_3$ and $x_1,x_2\in\mathcal{B}(x_0,\eta)\cap\Omega_3$.

Finally, let $x_0\in\Omega_2$ and $x\in\mathcal{B}(x_0,\eta)\cap\Omega_2$. Since $\|x\|\le\bar\delta/c_8$, the integrand of \eqref{equ:Tdef} equals $2/\beta(sx)$ for every $s\in[0,1]$. Adding and subtracting $1/\beta(x_0)$ in the integrand and splitting the integral at $\bar s:=\|x_0\|/\|x\|$, which may exceed $1$, we obtain
$$
\begin{aligned}
\dfrac{T(x)}{2}
=&\,x\int_0^{1} \left(\dfrac{1}{\beta(sx)}-\dfrac{1}{\beta(x_0)}+\dfrac{1}{\beta(x_0)}\right)\mathrm{d}s\\
=&\,\dfrac{x}{\beta(x_0)}
+\dfrac{x}{\|x\|}\underbrace{
\|x_0\|\int_0^{1} \left(\dfrac{1}{\beta(sx_0)}-\dfrac{1}{\beta(x_0)}\right)\mathrm{d}s}_{c_{9}(x_0)}\\
&\,+\underbrace{x
\int_{\bar s}^{1} \left(\dfrac{1}{\beta(sx)}-\dfrac{1}{\beta(x_0)}\right)\mathrm{d}s}_{c_{10}(x,x_0)}
\end{aligned}
$$
where the integral over $[0,\bar s]$ is rewritten by substituting $s\mapsto\bar s s$ and using $\beta(\bar s sx)=\beta(sx_0)$, as $\|\bar s sx\|=\|sx_0\|$. By the monotonicity of $1/\beta(\cdot)$, $c_{9}(x_0)\ge 0$. For $s$ in the closed interval $\mathcal{I}$ between $\bar s$ and $1$, $\|sx\|$ lies between $\|x_0\|=\bar s\|x\|$ and $\|x\|$. Hence, as $x,x_0\in\Omega_2$, $\|sx\|\ge\underline\delta/c_8$ and $\big|\|sx\|-\|x_0\|\big|\le\|x\||1-\bar s|=\big|\|x\|-\|x_0\|\big|\le\eta$. The Lipschitz bound on $1/\beta$ above thus gives
$$
\begin{aligned}
\|c_{10}(x,x_0)\|&{}\le\|x\||1-\bar s|\max_{s\in \mathcal{I}}\left|\dfrac{1}{\beta(sx)}-\dfrac{1}{\beta(x_0)}\right|\\
&{}\le c_{11}\eta^2.
\end{aligned}
$$
Then we have
\begin{equation}\label{equ:3Lip6}
\|T(x)/2-\underbrace{\left(\dfrac{x}{\beta(x_0)}+c_{9}(x_0)\dfrac{x}{\|x\|}\right)}_{\hat{T}(x)}\|\le c_{11}\eta^2.
\end{equation}

For $x_1,x_2\in\mathcal{B}(x_0,\eta)\cap\Omega_2$, let $z_i:=x_i/\beta(x_0)$, so that $\hat T(x_i)=z_i+c_{9}(x_0)z_i/\|z_i\|$. Following a similar argument to \eqref{equ:3Lip5} with $z_i$ in place of $x_i$ and $c_{9}(x_0)\ge0$ in place of $\bar\rho_3$, we have
\begin{equation}\label{equ:3Lip7}
\|\hat T(x_2)-\hat T(x_1)\|\ge\|x_2-x_1\|/\beta(x_0).
\end{equation}

By the triangle inequality, \eqref{equ:3Lip7} and \eqref{equ:3Lip6} with $x=x_1,x_2$ give $\|x_1-x_2\|/\beta(x_0)\le\|T(x_1)-T(x_2)\|/2+2c_{11}\eta^2$. Together with the first display of this proof, this yields
$$
\dfrac{\|x_1-x_2\|}{\beta_\eta(x_0)}\le\dfrac{\|T(x_1)-T(x_2)\|}{2}+4c_{11}\eta^2.
$$
Since $4c_{11}=c_{12}$, \eqref{equ:3Lip4} holds with $b=1/2$ for any $x_0\in\Omega_2$.
\end{proof}

\subsubsection{Proof of Lemma \ref{lem:transformation}}\label{sec:proof_transformation}
Let $x\in\mathcal{X}_\text{inv}$, $y=T(x)$ and denote the retracted value in \eqref{equ:transformation} by $u^\epsilon_\text{mpc}(x) := (1-\epsilon\beta(x)/d(\mathcal{U}))u_\text{mpc}(x)$. By \eqref{equ:trans_recover_first}, $\beta(x)>0$, \eqref{equ:transformation} and the hypothesis of the lemma,
\begin{align*}
\|u_\text{nn}(x)-u^\epsilon_\text{mpc}(x)\|
&=\|\widetilde{u}_\text{nn}(y)\beta(x) - u^\epsilon_\text{mpc}(x)\| \\
&=\beta(x) \left\| \widetilde{u}_\text{nn}(y) - \frac{u^\epsilon_\text{mpc}(x)}{\beta(x)} \right\| \\
&=\beta(x) \|\widetilde{u}_\text{nn}(y) - \widetilde{u}_\text{mpc}(y)\| \\
&\le \beta(x) \epsilon.
\end{align*}
Moreover, as $u_\text{mpc}(x)\in\mathcal{U}$, $\|u^\epsilon_\text{mpc}(x)-u_\text{mpc}(x)\|=\epsilon\beta(x)\|u_\text{mpc}(x)\|/d(\mathcal{U})\le\epsilon\beta(x)D(\mathcal{U})/d(\mathcal{U})$, which is the bound of Lemma \ref{lem:3Lip2} for the projection onto $\mathcal{U}'(\epsilon\beta(x))$. Hence,
\begin{align*}
&\|u_\text{nn}(x)-u_\text{mpc}(x)\| \\
&\le \|u_\text{nn}(x)-u^\epsilon_\text{mpc}(x)\| + \|u^\epsilon_\text{mpc}(x)-u_\text{mpc}(x)\| \\
&\le \beta(x)\epsilon + \epsilon\beta(x) {D(\mathcal{U})}/{d(\mathcal{U})} \le \beta(x)\bar\delta = \delta_2(x),
\end{align*}
where the last inequality uses $\epsilon\le\bar\delta d(\mathcal{U})/(2D(\mathcal{U}))$ and $d(\mathcal{U})\le D(\mathcal{U})$.
As noted after \eqref{equ:transformation}, $u^\epsilon_\text{mpc}(x)\in\mathcal{U}'(\epsilon\beta(x))$. Since $u_\text{nn}(x)-u^\epsilon_\text{mpc}(x)\in\mathcal{B}(\epsilon\beta(x))$ by the first display, it follows from \eqref{tightset} that $u_\text{nn}(x) \in \mathcal{U}$.
\qed

\subsubsection{Proof of Lemma \ref{lem:3Lip6}}\label{sec:proof_3Lip6}
By Lemma \ref{lem:T_bijective}, $T(\cdot)$ maps the ball $\mathcal{X}_D\supseteq\mathcal{X}_\text{inv}$ onto the ball $T(\mathcal{X}_D)=\mathcal{B}(\rho(R_D))$. Moreover, $T(\cdot)$ is a bijection of $\mathbb{R}^{n_x}$ onto itself, which is the first claim. Indeed, by Appendix \ref{sec:proof_T}, $T(0)=0$, $T(x)=\rho(\|x\|)x/\|x\|$ for $x\neq0$, and $\rho$ is continuous and strictly increasing with $\rho(0)=0$ and $\rho(s)\ge s$. Hence $\rho$ maps $[0,\infty)$ bijectively onto itself, and $T(\cdot)$ maps each sphere of radius $s$ bijectively onto the sphere of radius $\rho(s)$. Cases 2 and 3 below join points by segments, which may leave $T(\mathcal{X}_\text{inv})$. We therefore extend $u_\text{mpc}(\cdot)$ from $\mathcal{X}_\text{inv}$ to $\mathcal{X}_D$ with the same GLC $L_u$ by Kirszbraun's theorem, denote the extension again by $u_\text{mpc}(\cdot)$, and define $\widetilde{u}_\text{mpc}(\cdot)$ on $T(\mathcal{X}_D)$ by \eqref{equ:transformation}. It suffices to prove $\mathcal{L}(\widetilde{u}_\text{mpc},T(\mathcal{X}_D))\le L_u$, because the extension leaves $\widetilde{u}_\text{mpc}(\cdot)$ unchanged on $T(\mathcal{X}_\text{inv})$.

The ball $\mathcal{X}_D$ is the union of the three regions $\Omega_1,\Omega_2,\Omega_3$ defined in Appendix \ref{sec:proof_T}, on which, by \eqref{equ:relative_error}, $\beta(x)$ equals $\underline\delta/\bar\delta$, $c_8\|x\|/\bar\delta$ and $1$, respectively.

By \eqref{equ:transformation}, $\widetilde{u}_\text{mpc}(T(x))=g(x)u_\text{mpc}(x)$ for $x\in\mathcal{X}_D$, where the scalar $g(x):=\big(1-\epsilon\beta(x)/d(\mathcal{U})\big)/\beta(x)=1/\beta(x)-\epsilon/d(\mathcal{U})$ combines the retraction and the output scaling. Since $0\le\epsilon\beta(x)<d(\mathcal{U})$ under the condition on $\epsilon$ in Lemma \ref{lem:transformation}, we have $0<g(x)\le 1/\beta(x)$.

\textit{Case 1:} For any $y_1,y_2\in T(\Omega_1)$ or $y_1,y_2\in T(\Omega_3)$, we have
$$
\begin{aligned}
&\|\widetilde{u}_\text{mpc}(y_1) - \widetilde{u}_\text{mpc}(y_2)\| \\
& {}\le L_u\|T^{-1}(y_1) - T^{-1}(y_2)\|/\beta(T^{-1}(y_1)) \\
& {}\le L_u\|y_1 - y_2\|
\end{aligned}
$$
where the first inequality holds because $g(\cdot)$ is constant on $\Omega_1$ and on $\Omega_3$, $0<g\le1/\beta$ and $u_\text{mpc}(\cdot)$ has GLC $L_u$ on $\mathcal{X}_D$. On $T(\Omega_1)=\Omega_1'$, the last inequality holds because $\beta(T^{-1}(y_1))=\underline\delta/\bar\delta$ and $T^{-1}(y)=T_1^{-1}(y)=\underline\delta y/(2\bar\delta)$. On $T(\Omega_3)$, it holds because $\beta(T^{-1}(y_1))=1$ and $T^{-1}(\cdot)$ is $1$-Lipschitz continuous by Lemma \ref{lem:T_bijective}.
Thus, $\mathcal{L}(\widetilde{u}_\text{mpc},T(\Omega_i)) \leq L_u$, $i=1,3$.

\textit{Case 2:} Fix $\eta\in(0,\eta_\text{max})$ with $\eta_\text{max}$ from Lemma \ref{lem:3Lip4}, and let $x_0\in\Omega_2$ and $x_1,x_2\in\mathcal{B}(x_0,\eta)\cap\Omega_2$. Then
$$
\begin{aligned}
&\|g(x_1)u_\text{mpc}(x_1)-g(x_2)u_\text{mpc}(x_2)\|\\
&\le \|u_\text{mpc}(x_1)\||g(x_1)-g(x_2)|+g(x_2)\|u_\text{mpc}(x_1)-u_\text{mpc}(x_2)\|\\
&\le L_u\|x_1\|\left|\dfrac{1}{\beta(x_1)}-\dfrac{1}{\beta(x_2)}\right|+L_u\|x_1-x_2\|/\beta(x_2),
\end{aligned}
$$
where the second inequality uses $g(x_2)\le1/\beta(x_2)$, $g(x_1)-g(x_2)=1/\beta(x_1)-1/\beta(x_2)$ and $\|u_\text{mpc}(x_1)-u_\text{mpc}(0)\|\le L_u\|x_1\|$ with $u_\text{mpc}(0)=0$ by the proof of Lemma \ref{lem:3Lip}, as $0\in\mathcal{X}_\text{inv}$. Since $1/\beta(x)=\bar\delta/(c_8\|x\|)$ on $\Omega_2$ and $\big|\|x_2\|-\|x_1\|\big|\le\|x_1-x_2\|$, we have
$$
\|x_1\|\left|\dfrac{1}{\beta(x_1)}-\dfrac{1}{\beta(x_2)}\right|
=\dfrac{\bar\delta}{c_8}\dfrac{\big|\|x_2\|-\|x_1\|\big|}{\|x_2\|}\le\dfrac{\|x_1-x_2\|}{\beta(x_2)}.
$$
Thus,
$$
\begin{aligned}
&\|g(x_1)u_\text{mpc}(x_1)-g(x_2)u_\text{mpc}(x_2)\|\\
&\le 2L_u\|x_1-x_2\|/\beta_\eta(x_0)\\
&\le 2L_u\left(\|T(x_1)/2-T(x_2)/2\|+c_{12}\eta^2\right)\\
&= L_u\|T(x_1)-T(x_2)\|+2L_uc_{12}\eta^2,
\end{aligned}
$$
where the first step uses $\beta(x_2)\ge\beta_\eta(x_0)$ and the second step follows from Lemma \ref{lem:3Lip4} with $b=1/2$. Writing $y_i=T(x_i)$, we obtain
\begin{equation}\label{equ:3Lip9}
\|\widetilde{u}_\text{mpc}(y_1)-\widetilde{u}_\text{mpc}(y_2)\|
 \le L_u\|y_1-y_2\|+2L_uc_{12}\eta^2
\end{equation}
for any $y_1,y_2\in T(\Omega_2)$ whose pre-images lie in a common ball $\mathcal{B}(x_0,\eta)$ with $x_0\in\Omega_2$.

Now let $y_1,y_2\in T(\Omega_2)$ be such that the segment joining them lies in $T(\Omega_2)$. We select a series of intermediate points $y_1=\bar y_0,\bar y_1,\cdots,\bar y_{n_y}=y_2$ equally spaced on that segment, with $n_y=\lceil\|y_1-y_2\|/\eta\rceil$, such that the distance between any two consecutive points is no greater than $\eta$. Since $T^{-1}(\cdot)$ is $1$-Lipschitz continuous by Lemma \ref{lem:T_bijective}, $\|T^{-1}(\bar y_{i+1})-T^{-1}(\bar y_i)\|\le\eta$. Hence \eqref{equ:3Lip9} applies to each consecutive pair with $x_0=T^{-1}(\bar y_i)\in\Omega_2$, and we have
$$
\begin{aligned}
&\|\widetilde{u}_\text{mpc}(y_1)-\widetilde{u}_\text{mpc}(y_2)\|\\
&\le \sum_{i=0}^{n_y-1}\|\widetilde{u}_\text{mpc}(\bar y_{i+1})-\widetilde{u}_\text{mpc}(\bar y_i)\|\\
&{}\le \sum_{i=0}^{n_y-1}\left(L_u\|\bar y_{i+1}-\bar y_i\|+2L_uc_{12}\eta^2\right)\\
&{}\le L_u\|y_1-y_2\|+2L_uc_{12}\eta^2 \lceil\|y_1-y_2\|/\eta\rceil\\
&{}\le L_u\|y_1-y_2\|+c_{13}\eta
\end{aligned}
$$
where the last step uses $\lceil\|y_1-y_2\|/\eta\rceil\le\|y_1-y_2\|/\eta+1$, $\|y_1-y_2\|\le 2\rho(R_D)$ and $\eta<\eta_\text{max}$, and $c_{13}{}:=2L_uc_{12}(2\rho(R_D)+\eta_\text{max})$ is a constant independent of $\eta$. Letting $\eta\downarrow 0$ gives $\|\widetilde{u}_\text{mpc}(y_1)-\widetilde{u}_\text{mpc}(y_2)\|\le L_u\|y_1-y_2\|$.

\textit{Case 3:} Let $y_1,y_2\in T(\mathcal{X}_D)$. Since $T(\cdot)$ maps each sphere of radius $s$ onto the sphere of radius $\rho(s)$ and $\rho$ is continuous and increasing, $T(\Omega_1)$ is the ball of radius $\rho(\underline\delta/c_8)$, while $T(\Omega_2)$ and $T(\Omega_3)$, if non-empty, are the closed shells between the radii $\rho(\underline\delta/c_8)$ and $\rho(\bar\delta/c_8)$ and between $\rho(\bar\delta/c_8)$ and $\rho(R_D)$, respectively. The segment joining $y_1$ and $y_2$ lies in the ball $T(\mathcal{X}_D)$ and meets the spheres of radii $\rho(\underline\delta/c_8)$ and $\rho(\bar\delta/c_8)$ in at most four points, which split it into sub-segments whose lengths sum to $\|y_1-y_2\|$. In the interior of each sub-segment, $\|y\|$ differs from both radii, so, by continuity, the sub-segment lies in a single $T(\Omega_i)$ and Case 1 or Case 2 applies to its end points. The triangle inequality then gives $\|\widetilde{u}_\text{mpc}(y_1)-\widetilde{u}_\text{mpc}(y_2)\|\le L_u\|y_1-y_2\|$. Thus, \eqref{equ:3Lip8} holds.
\qed

\subsubsection{Proof of Theorem \ref{thm:main2}}\label{sec:proof_thm2}
Convergence follows the proof of Lemma \ref{lem:convergence}, with separate bounds for states far from and near the origin, while feasibility and the network size follow the proof of Theorem \ref{thm:main}.

Let the ReLU network satisfy the hypothesis of Lemma \ref{lem:transformation}, so that \eqref{equ:2con2} and \eqref{equ:2con1} hold. Write $x^+:=Ax+Bu_\text{nn}(x)$.

\textit{Stage 1 (non-uniform error regime):} Let $x\in\mathcal{X}_\text{inv}$ with $\|x\|\ge\underline\delta/c_8$. By \eqref{equ:relative_error} and \eqref{equ:2con1}, $\|u_\text{nn}(x)-u_\text{mpc}(x)\|\le c_8\|x\|$. As the error bound enters \eqref{equ:lemma2-part} only at the current state, \eqref{newlyapunova} holds at $x$ with $\delta_1=c_8\|x\|$ whenever $x^+$ is feasible, as Lemma \ref{lem:3Lip} requires. The computation of Section \ref{sec:motivation} then gives $v_\text{mpc}(x^+)\le\lambda v_\text{mpc}(x)$ with $\lambda:=1-c_3/c_2+(c_4c_8+c_5c_8^2)/c_1$. By \eqref{equ:stability_condition} and $c_3\le c_2$ from the proof of Lemma \ref{lem:convergence}, $\lambda\in[0,1)$.

\textit{Stage 2 (uniform error regime):} Let $x\in\mathcal{X}_\text{inv}$ with $\|x\|\le\underline\delta/c_8$. By \eqref{equ:relative_error} and \eqref{equ:2con1}, $\|u_\text{nn}(x)-u_\text{mpc}(x)\|\le\underline\delta$, so \eqref{newlyapunova} holds at $x$ with $\delta_1=\underline\delta$ whenever $x^+$ is feasible. As in the proof of Lemma \ref{lem:convergence}, this gives $v_\text{mpc}(x^+)\le\psi(v_\text{mpc}(x))$ with $\psi(\cdot)$ defined there for $\delta_1=\underline\delta$. Since $\underline\delta\le\bar\delta<c_7\sqrt{\gamma}$ by \eqref{equ:delta1}, i.e., $c_6\underline\delta^2<\gamma$, that proof also shows that $\psi(\cdot)$ is non-decreasing with $\psi(\gamma)<\gamma$ and fixed point $a_\infty=c_6\underline\delta^2$.

\textit{Combining the two regimes.} Let $\Phi(w):=\max\{\lambda w,\psi(w)\}$, $w\ge0$. As the two stages cover $\mathcal{X}_\text{inv}$, $v_\text{mpc}(x^+)\le\Phi(v_\text{mpc}(x))$ whenever $x^+$ is feasible. Since $0\le\lambda<1$, the properties of $\psi(\cdot)$ in Stage 2 carry over to $\Phi(\cdot)$: it is non-decreasing, $\Phi(a_\infty)=a_\infty$ and $\Phi(\gamma)<\gamma$. The two stages use $u_\text{nn}(x)$ only through \eqref{equ:2con1}, which every input on the segment between $u_\text{mpc}(x)$ and $u_\text{nn}(x)$ also satisfies. Hence the feasibility argument in the proof of Theorem \ref{thm:main}, with $\psi(\cdot)$ replaced by $\Phi(\cdot)$, shows that $x^+$ is feasible. Therefore $v_\text{mpc}(x^+)\le\Phi(v_\text{mpc}(x))\le\Phi(\gamma)<\gamma$, i.e., \eqref{equ:2con3} holds.

Let $x_0\in\mathcal{X}_\text{inv}$. By \eqref{equ:2con3}, $x_k\in\mathcal{X}_\text{inv}$ for all $k$, so induction with the monotonicity of $\Phi(\cdot)$ gives $v_\text{mpc}(x_k)\le\Phi^k(\gamma)$. For $w\ge a_\infty$, $\lambda w-a_\infty\le\lambda(w-a_\infty)$ and, by the proof of Lemma \ref{lem:convergence}, $\psi(w)-a_\infty\le q(w-a_\infty)$ with $q\in[0,1)$ defined after that lemma. Hence $\Phi(w)-a_\infty\le\max\{\lambda,q\}(w-a_\infty)$. As $\Phi(\cdot)$ maps $[a_\infty,\gamma]$ into itself, iterating from $w=\gamma$ yields $v_\text{mpc}(x_k)-c_6\underline\delta^2\le\max\{\lambda,q\}^k\gamma$, where $\max\{\lambda,q\}<1$ by Stage 1 and the definition of $q$. Thus, the closed-loop system converges exponentially to $\mathcal{V}(c_6\underline\delta^2)\subseteq\mathcal{X}_\text{inv}$. This set is invariant, as $v_\text{mpc}(x^+)\le\Phi(v_\text{mpc}(x))\le\Phi(a_\infty)=a_\infty$ for each of its points $x$. By \eqref{equ:exp_stability2}, any state within this set satisfies $\|x\| \le \sqrt{c_6/c_1} \cdot \underline\delta$. The condition $\underline\delta \le \sqrt{c_1/c_6} d(\mathcal{X}_\text{inv}')$ from the definition of $\underline\delta$ gives $\mathcal{B}(\sqrt{c_6/c_1}\,\underline\delta)\subseteq\mathcal{B}(d(\mathcal{X}_\text{inv}'))\subseteq\mathcal{X}_\text{inv}'$, as $\mathcal{X}_\text{inv}'$ is closed. Hence the invariant set $\mathcal{V}(c_6\underline\delta^2)$ lies in $\mathcal{X}_\text{inv}'$. Finally, \eqref{equ:exp_stability2} gives $\|x_k\|\le\sqrt{c_6/c_1}\,\underline\delta+\max\{\lambda,q\}^{k/2}\sqrt{\gamma/c_1}$. By the ball inclusion above, $\mathrm{dist}(x_k,\mathcal{X}_\text{inv}')\le\max\{\lambda,q\}^{k/2}\sqrt{\gamma/c_1}$, i.e., $x_k$ converges exponentially to $\mathcal{X}_\text{inv}'$.

The network complexity bound is derived by adapting the proof of Theorem \ref{thm:main} with three substitutions. First, the approximation target is $\widetilde{u}_\text{mpc}(\cdot)$ on $T(\mathcal{X}_\text{inv})$ instead of $\Pi_{\mathcal{T}(\mathcal{U},\epsilon)}\circ u_\text{mpc}(\cdot)$ on $\mathcal{X}_\text{inv}$. Its GLC is still at most $L_u$ by Lemma \ref{lem:3Lip6}. By Lemma \ref{lem:T_bijective}, $T(\cdot)$ is a homeomorphism of $\mathcal{X}_D\supseteq\mathcal{X}_\text{inv}$ onto $T(\mathcal{X}_D)$, and, as $\|T(x)\|=\rho(\|x\|)$ with $\rho$ strictly increasing, it maps the interior of $\mathcal{X}_D$ onto that of $T(\mathcal{X}_D)$. Hence $T(\mathcal{X}_\text{inv})$ is, like $\mathcal{X}_\text{inv}$, compact and the closure of its interior. Second, by Lemma \ref{lem:3Lip5}, the domain radius $D(\mathcal{X}_\text{inv})$ becomes $D(T(\mathcal{X}_\text{inv}))\le D_2$. Third, the tightening parameter is \eqref{equ:eps} with $\delta_1$ replaced by $\bar\delta$, as admitted by Lemma \ref{lem:transformation}. The adapted argument then shows that, under \eqref{equ:2main}, a network with width \eqref{equ:2width} and depth \eqref{equ:2depth} satisfies the hypothesis of Lemma \ref{lem:transformation} assumed above.
\qed

\subsection{ReLU NN Approximation}\label{sec:nn}
For the reader's convenience, we restate the following lemma on ReLU NN approximation from \cite{shen2022optimal}, in its version for bounded domains applied to the cube $[-D(\mathcal{E}),D(\mathcal{E})]^{n_x}\supseteq\mathcal{E}$.
\begin{lemma}[{\cite[Theorem 2.5]{shen2022optimal}}]\label{lem:3nn}
Given a bounded set $\mathcal{E}\subseteq\mathbb{R}^{n_x}$ containing the origin with $D(\mathcal{E})>0$ and a bounded continuous function $f \in C(\mathcal{E})$, for any $L',N' \in \mathbb{N}_+$, and for $p \in [1, \infty]$, there exists a ReLU NN $\phi$ with width $C_1\max\{{n_x}\lfloor N'^{1/{n_x}}\rfloor, N'+2\}$ and depth (number of hidden layers) $11L' + C_2$, such that
\begin{equation}\label{equ:3nn}
\begin{aligned}
\|f - \phi\|_{L^p(\mathcal{E})} \leq{}& 131(2D(\mathcal{E}))^{n_x/p}\sqrt{{n_x}}\\
&\times\omega^{\mathcal{E}}_f\big(2D(\mathcal{E})(N'^2L'^2\log_3(N'+2))^{-1/{n_x}}\big),
\end{aligned}
\end{equation}
where $\omega^{\mathcal{E}}_f(s):=\sup\{|f(x)-f(y)|:x,y\in\mathcal{E},\|x-y\|\le s\}$ is the modulus of continuity of $f$ on $\mathcal{E}$; if $p \in [1, \infty)$, then $C_1 = 16$ and $C_2 = 18$; if $p = \infty$, then $C_1 = 3^{{n_x}+3}$ and $C_2 = 18 + 2{n_x}$.
\end{lemma}


\begin{IEEEbiography}
[{\includegraphics[width=1in,height=1.25in,clip,keepaspectratio]{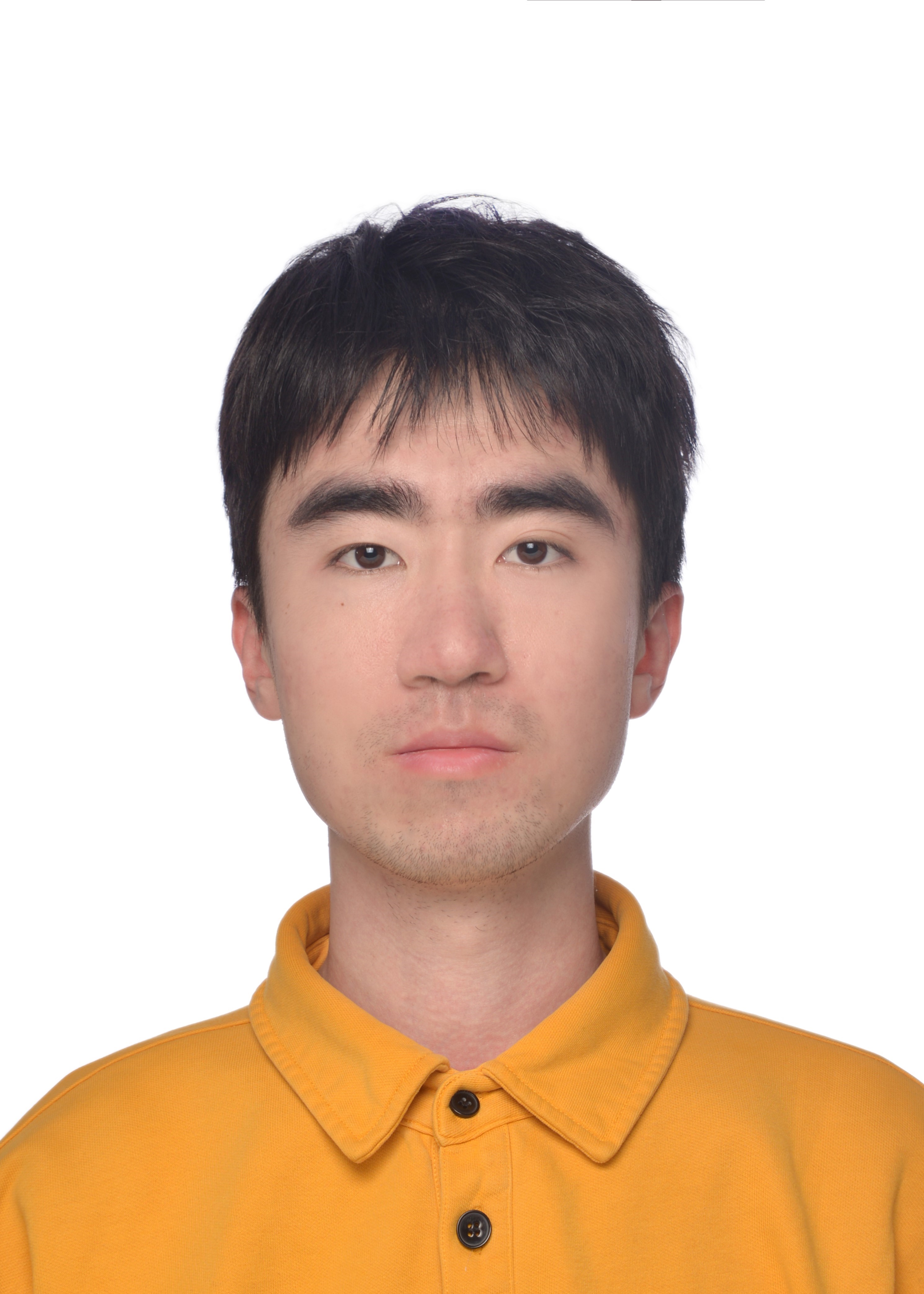}}]
{Xingchen Li} (Student Member, IEEE) received the B.S. degree from the Department of Automation, Tsinghua University, Beijing, China, in 2021. He is currently pursuing the Ph.D. degree at the Department of Automation, Tsinghua University, Beijing, China. His research interests lie at the intersection of machine learning, optimization, and control theory. He focuses on developing and analyzing learning-based methods for fast and reliable decision-making, including the use of GPU acceleration for high-performance optimization and control.\vspace{-12pt}
\end{IEEEbiography}

\begin{IEEEbiography}[{\includegraphics[width=1in,height=1.25in,clip,keepaspectratio]{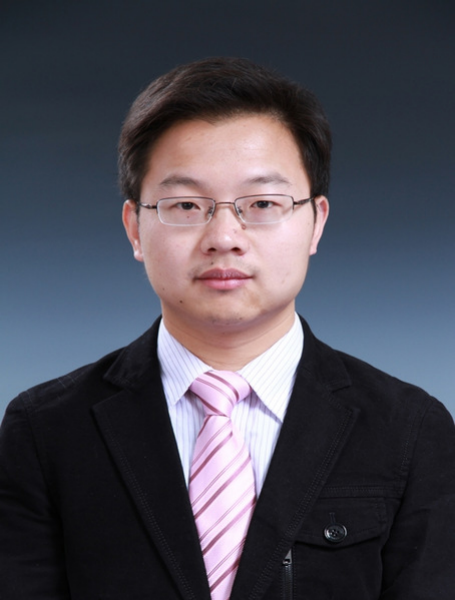}}]
{Keyou You} (Senior Member, IEEE) received the B.S. degree in Statistical Science from Sun Yat-sen University, Guangzhou, China, in 2007 and the Ph.D. degree in Electrical and Electronic Engineering from Nanyang Technological University (NTU), Singapore, in 2012. After briefly working as a Research Fellow at NTU, he joined Tsinghua University in Beijing, China, where he is now a Full Professor in the Department of Automation. He held visiting positions at Politecnico di Torino, the Hong Kong University of Science and Technology, the University of Melbourne, etc.

Prof. You’s research interests focus on the intersections between control, optimization and learning as well as their applications in autonomous systems. He received the Guan Zhaozhi Award at the 29th Chinese Control Conference in 2010 and the ACA (Asian Control Association) Temasek Young Educator Award in 2019. He received the National Science Funds for Excellent Young Scholars in 2017, and for Distinguished Young Scholars in 2023. Currently, he is an Associate Editor for {\em Automatica}, {\em IEEE Transactions on Automatic Control}, and {\em IEEE Transactions on Control of Network Systems}.
\end{IEEEbiography}

\end{document}